\documentclass{tMAM2e}

\usepackage[utf8]{inputenc}
\usepackage{amsmath}
\usepackage{amssymb}
\usepackage{hyperref}
\usepackage{url}
\usepackage{graphicx}
\usepackage{booktabs}
\usepackage{lipsum}
\usepackage{babel}
\usepackage{algorithm}
\usepackage{algorithmic}
\usepackage{xcolor}
\usepackage{indentfirst}
\usepackage{natbib}[citations.bib]
\usepackage{listings}
\usepackage{bm}
\usepackage{bbm}
\usepackage{float}
\usepackage{mathtools}
\usepackage{xcolor}
\usepackage{array}   
\usepackage{wasysym}
\usepackage{booktabs}
\usepackage{lineno}
\usepackage{amsthm}
\usepackage{tikz}
\usepackage{enumitem}
\usepackage{multirow}
\usepackage{caption}
\usepackage{subcaption}
\usepackage{svg}
\usepackage{lscape}
\usepackage{changepage}

\usepackage{romannum}

\modulolinenumbers[5]
\DeclareUnicodeCharacter{202F}{\,}

\usepackage{multicol}
\usepackage{graphicx}
 
\newcommand{\RNum}[1]{\uppercase\expandafter{\romannumeral #1\relax}}

\def\barroman#1{\sbox0{#1}\dimen0=\dimexpr\wd0+1pt\relax
  \makebox[\dimen0]{\rlap{\vrule width\dimen0 height 0.06ex depth 0.06ex}%
    \rlap{\vrule width\dimen0 height\dimexpr\ht0+0.03ex\relax 
            depth\dimexpr-\ht0+0.09ex\relax}%
    \kern.5pt#1\kern.5pt}}

\newcommand{\chromas}{\chi^{\text{\barroman{V}}}}
\newcommand{\notes}{\mathcal{N}^{\text{\barroman{V}}}}
\newcommand{\notesa}{\mathcal{N}^{\text{\barroman{V}}}_{\alpha}}
\newcommand{\notese}{\nu_{i,j}}

\newcommand{\backdelta}{|\delta^{-1}|}

\newcommand{\interp}{\mathcal{I}}

\newtheorem{theorem}{Theorem}[section]
\newtheorem{definition}{Definition}[theorem]
\newtheorem{corollary}{Corollary}[theorem]
\newtheorem{lemma}{Lemma}[theorem]
\newtheorem*{remark}{Remark}
\newtheorem{proposition}[theorem]{Proposition}
\newtheorem*{example}{Example}

\newtheoremstyle{case}{}{}{\itshape}{\parindent}{\bfseries}{:}{ }{}
\theoremstyle{case}
\newtheorem{case}{Case}

\begin{document}

\title{A Geometric Framework for Pitch Estimation on Acoustic Musical Signals}

\author{
Tom Goodman$^{*}$\thanks{$^{*}$\text{School of Computer Science, University of Birmingham}},
\text{Karoline van Gemst}$^{\dagger}$\thanks{$^{\dagger}$ \text{School of Mathematics and Statistics, University of Sheffield}},\text{and Peter Tino}$\protect^{*}$}

\maketitle

\begin{abstract}
This paper presents a geometric approach to pitch estimation (PE) - an important problem in Music Information Retrieval (MIR), and a precursor to a variety of other problems in the field. Though there exist a number of highly-accurate methods, both mono-pitch estimation and multi-pitch estimation (particularly with unspecified polyphonic timbre) prove computationally and conceptually challenging. A number of current techniques, whilst incredibly effective, are not targeted towards eliciting the underlying mathematical structures that underpin the complex musical patterns exhibited by acoustic musical signals. Tackling the approach from both a theoretical and experimental perspective, we present a novel framework, a basis for further work in the area, and results that (whilst not state of the art) demonstrate relative efficacy. The framework presented in this paper opens up a completely new way to tackle PE problems, and may have uses both in traditional analytical approaches, as well as in the emerging machine learning (ML) methods that currently dominate the literature. 

\end{abstract}


\section{Introduction}
Music Information Retrieval is {increasingly gaining momentum} as a cross-disciplinary field of research \citep{muller2011signal}, pulling together techniques and researchers from computer science, cognitive science, musicology, and electrical engineering, amongst others. Despite this, many problems that on the surface appear to be trivially solvable from a human perspective have proved intractable for computers, and thus, have remained unsolved. 

Pitch estimation is one such problem - the ability to take a musical signal as an input, and at any given position in the signal, be able to ascertain what notes (pitch chroma/pitch height pairs) are present. This is made difficult because of the sheer volume of timbres that could be present, and notes that could be played with various amplitudes, amongst other things. The remarkable variety that exists within music, especially with increasing levels of polyphony, render the problem of multi-pitch estimation (MPE) incredibly challenging indeed. 

Many approaches to pitch estimation are restricted to specific instruments (for example, guitar tuning devices) \citep{bock2012polyphonic, steinberger1996chromatic, schramm2017multi}. As a result, researchers are able to exploit certain properties and assumptions (pertaining to the specific context in which the methods will operate) to increase the accuracy of their approaches. This helps to circumvent the greater difficulty of developing approaches that function from a more generalised perspective.

Relatively recently, there has been a notable influx of approaches to problems in Music Information Retrieval that utilise state-of-the-art machine learning techniques \citep{8683582, wu2018automatic}. Whilst the accuracy of such approaches has proved good, the lack of ability to inspect (and further, understand), the inner workings of them has resulted in a lack of deep insight - especially into the underlying mathematical structures that they are approximating. This provides a strong motivation to develop novel algorithmic approaches \citep{goodman2018real}, in an attempt to discern the intrinsic models - perhaps even those that we, as humans, exploit to perform these tasks. 

Additionally, given the growing use of processing power for these approaches \citep{strubell2019energy}, and the increasingly pressing nature of climate change - with the IPCC special report on global warming of 1.5°C concluding that our CO${2}$ emissions need to be halved by 2030 \citep{djalante2019icaa} to curtail the increasing numbers of natural disasters - it is imperative that we, as researchers, consider the efficiency of our algorithms and approaches. With the era of Moore's Law waning \citep{leiserson2020there}, it is ever important to work towards algorithmic approaches that may prove more computationally efficient. Indeed, although the cloud may appear to be some ephemeral resource, somewhere there are real computers using real electricity, and having real and profound impacts - both financially, and environmentally. Looking again to model acoustic musical signals (and more generally, AI problems) from mathematical points of view could give rise to novel and efficient approaches that, whilst less efficacious in the short-term than their machine-learning counterparts, may have the potential to open up new branches of research in the long-term, with a particular focus on understanding the problems being solved. 

In this paper, a novel geometric perspective and methodology for both mono- and multi-pitch estimation is presented. 
Section~\ref{sec:related} gives a broad overview of the related work, and a number of inspirations for the paper, following which Section~\ref{sec:reaching} introduces the model, which is then formalised in Section~\ref{sec:proposed}. Building on this, Section~\ref{sec:edge} examines the `edge cases' that arise in more detail. Sections~\ref{sec:reduction} and~\ref{sec:prevalence} present an approach to examine the prevalence of edge cases, and apply it to data sampled randomly from the total space. Moving on from the more theoretical perspective, Section~\ref{sec:real} applies the model to `real-world' data, with Section~\ref{sec:eval} more closely analysing the performance of simple algorithms working on the proposed model. Finally, Section~\ref{sec:future} proposes future direction for the work, and a number of potential applications of the framework. 

\section{Related Work}\label{sec:related}
\begin{figure}
    \centering
    \includegraphics[scale=3.0]{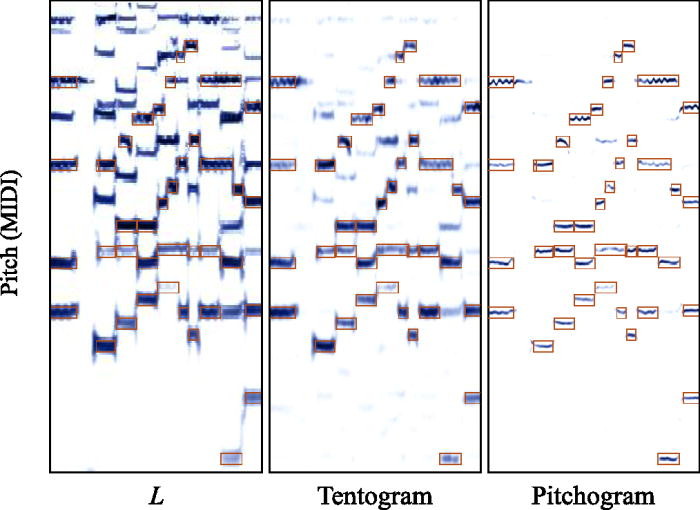}
    \caption{Demonstration of the progression from Spectogram to Tentogram, and then to Pitchogram \citep{elowsson2020polyphonic}.}
    \label{fig:elowsson_fig}
\end{figure}

Over the past few years, a plethora of approaches have been taken to tackle the challenge of pitch estimation - with methods utilising the wavelet transform \citep{kumar2020wavelet}, two-dimensional spectra \citep{zhang2020multi}, and visual infomation (i.e. by viewing the physical instrument itself) \citep{koepke2019visual}, amongst others. Though varied, much of the cutting-edge research is reliant on machine learning (ML) techniques - not necessarily seeking to better-understand the underlying structures present, and opting rather to maximise efficacy of the respective approaches.

Elowsson \citep{elowsson2020polyphonic} proposed a method for MPE that relies on `deep layered learning' \citep{elowsson2014polyphonic, elowsson2018deep}. It uses a multi-stage system of neural networks and processing steps to elicit pitch contours - i.e. pitch information coupled with an onset and offset for each distinct note. From the MPE side, they opted to create a `Tentogram' (i.e. a tentative spectogram) through a spectral summation (their Section IV-F), whitening, and logistic regression, which provided a much cleaner basis from which to detect pitch peaks. A neural network was then used to convert the Tentogram into a `Pitchogram', using parabolic interpolation to achieve a 1 cent resolution. From the Pitchogram, `blobs' are then identified \citep{miron2014audio}, with subsequent regions then merged (where related), and finally a peak ridge (1D contour) was extracted. This method exhibited state-of-the-art performance on the MAPS \citep{emiya2010maps}, Bach10 \citep{duan2015bach10}, TRIOS, and MIREX Woodwind quintet datasets.

Kelz, Rainer and B{\"o}ck looked to derive pitch contours from polyphonic audio specifically from piano \citep{kelz2019deep}. Previous ML approaches tended to use many networks to extract the various features, but they instead opted to use a shared representation \citep{ngiam2011multimodal} to simultaneously predict the ADSR (Attack, Decay, Sustain, Release) aspects of each note. By approaching the problem in this way, they were able to increase the potential for generalisation to other instruments, as the representations that prove useful to this task can be adapted to others. Further, rather than opting to train another network to learn the ADSR envelopes of the input, they chose to handcraft a Hidden Markov Model (HMM) with states corresponding to each envelope, as well as an additional state to represent that a note is not currently sounding. Similarly to this approach, it could be possible to use the proposed framework (Section~\ref{sec:proposed}) in a kindred manner (albeit from a different perspective).

The autocorrelation function (ACF) of a signal $\{x(n)\}_{n=0}^{N-1}$ at lag $m$,
\begin{equation*}
    r_{xx}(m) = \sum\limits_{n=m}^{N-1}x(n)x(n-m), 
\end{equation*}
 has been widely used to elicit pitch information from time-domain signals \citep{rabiner1977use, amado2008pitch, kraft2015polyphonic}. Recently, de Obald{\'\i}a, and Z{\"o}lzer conducted a study looking at the efficacy of the ACF on non-stationary sounds (to extract the fundamental period), and presented a number of augmentations that allowed them to improve on the current state-of-the-art approaches for monophonic pitch estimation \citep{de2019improving}. They achieved this by utilising musicological knowledge to construct a heuristic that identifies non-related jumps in the pitch contour, and subsequently modifying the signal to compensate for these. They reported state-of-the-art results on both speech (including PTDB-TUG \citep{pirker2011pitch}) and musical signals (including Bach10).

\begin{figure}
    \centering
    \includegraphics{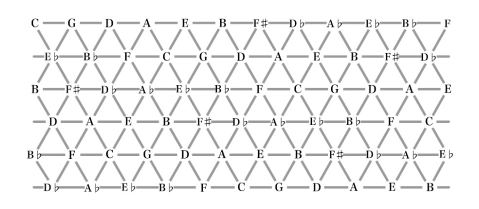}
    \caption{A Tonnetz without identified edges - clearly showing its periodic nature in both the vertical and horizontal directions (Image from \citep{bergstromimage}).}
    \label{fig:tonnetz}
\end{figure}
First described by Euler in 1739, the Tonnetz (Figure~\ref{fig:tonnetz}) \citep{euler19681739} presents a way to spatially demonstrate the relationship between chords. Each row corresponds to the circle of fifths, with each subsequent row corresponding to the previous one with each element shifted from position $n$, to position $(n + 3)\bmod 12$, and aligned such that the closest two notes diagonally `upwards' correspond to a major third in each direction. By identifying both vertically and horizontally, it is clear that the Tonnetz is in fact toroidal in nature. It has proven particularly useful in describing voice leadings in music - with distance between triangles (i.e. chords) on the Tonnetz corresponding to musical `distance' between chords. 

In addition to his work on the generalised Tonnetz \citep{tymoczko2012generalized}, Tymockzo posits a geometrical treatment of music theory in his book `A Geometry of Music' \citep{tymoczko2010geometry} - driven by underlying musicological knowledge. Where Lewin tackled this kind of formalisation from a group-theoretical perspective \citep{lewin2011generalized}, Tymockzo (whilst still basing his approach on symmetries) opted to employ a more visual approach; considering pitch/chord spaces in such a way that proves useful to composers and musicologists alike, without a necessarily profound understanding of the underpinning mathematics. 

Tymockzo defines musical objects, which are essentially ordered collections of notes (Eg. (C4, E4, G4)), and five ``OPTIC'' operations - 
\begin{itemize}
    \item \textbf{O}ctave - changing the pitch height of a note;
    \item \textbf{P}ermutation - reordering the object (which has the effect of changing which voice is playing which note);
    \item \textbf{T}ransposition - uniformly shifting all notes in an object by a given offset (and direction);
    \item \textbf{I}nversion - essentially reflection about a point in pitch space (i.e. pitches ordered chromatically along a 1D line);
    \item \textbf{C}ardinality change - introducing a new voice that duplicates a note that is already present in the object.
\end{itemize}
These describe transformations between musical objects. Further, he goes on to define a variety of musical constructs (such as chords and tone rows) in terms of the set of OPTIC transformations under which each construct remains invariant. 

Building on this framework, he defines a two-note chord space, on which progressions between two-note musical objects lie (Eg. (C4, E4)$\to$(C4, E$\flat$4)). By enumerating the whole space, and identifying the edges with a twist (which is necessary as, when enumerated fully, the vertical edges of the two-note chord space are the reverse of one another), the two-note chord space forms a M\"{o}bius strip. The use of this space in analysis is then demonstrated practically by applying it to elicit musical insights on pieces (such as Brahms' Op. 116, No. 5), that would otherwise be obscure if viewed in traditional notation. He goes on to provide a generalisation of n-note chord spaces in higher dimensions - with a three-note chord space forming a twisted triangular two-torus. 

Inspired in particular by Tymockzo's work, and historic algorithms such as Noll's Harmonic Pitch Spectrum (HPS) \citep{noll1970pitch} and De Cheveign{\'e} and Kawahara's YIN algorithm \citep{de2002yin}, this paper sets out to look at the problem of pitch estimation from a geometric point of view, and construct algorithms from the building blocks laid out herein. 

\section{Reaching a Model}\label{sec:reaching}
\begin{figure*}[]
    \centering
    \includegraphics[scale=0.8]{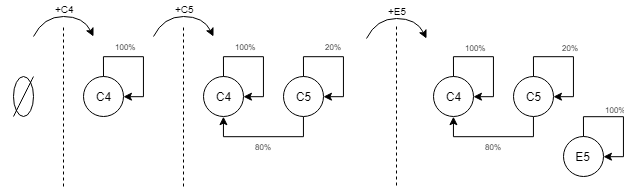}
    \caption{The build-up of a simplistic probabilistic representation.}
    \label{fig:prob}
\end{figure*}
Consider a frequency-sorted (low to high) set of notes, for example, \{(C, 4), (C, 5), (E, 5)\}. One can imagine wanting to build up some representation of where each note is likely to have originated. Figure~\ref{fig:prob} shows the iterative construction of one such model - which is somewhat adjacent conceptually to Markov chains. This could instead be viewed as a directed graph, with an edge from each node to every other node that it could potentially be a harmonic of. Weights are chosen from some node, B, to another node, A (of which B is potentially a harmonic) to be i, such that B is the i$^{\text{th}}$ harmonic of A (Figure~\ref{fig:graphgrid}). Here, each weight corresponds to some measure of the likelihood that a given note is generated by another. For example, given the presence of C4, the probability that C5 is a fundamental (i.e. generated by itself) may be 20\%, whereas there may be an 80\% probability that it is instead a harmonic of C4. These are, of course, toy values, and it is more than likely more realistic to readjust all weights following the addition of each note, but the underlying concept remains the same.

\begin{figure}
    \centering
    \includegraphics[scale=0.4]{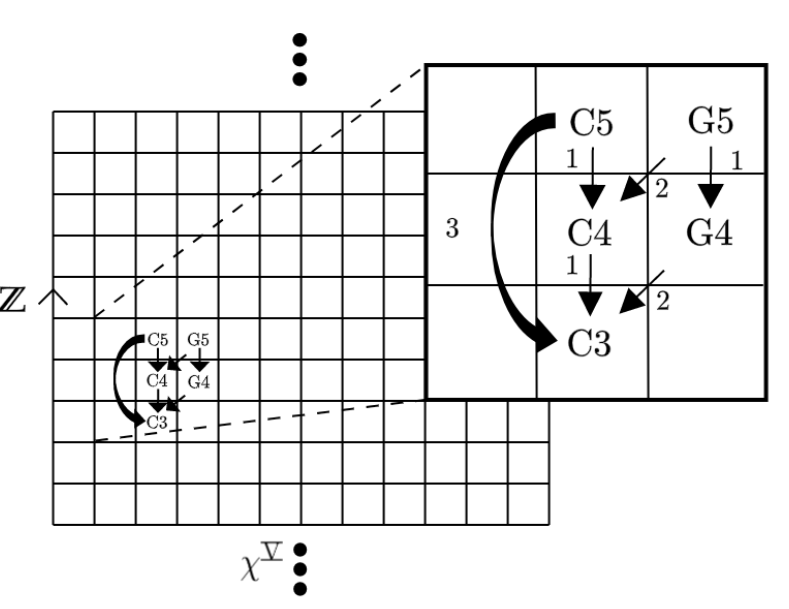}
    \caption{Visualisation of the graphical structure overlaid onto the grid.}
    \label{fig:graphgrid}
\end{figure}

\begin{figure}
    \centering
    \includegraphics[scale=0.4]{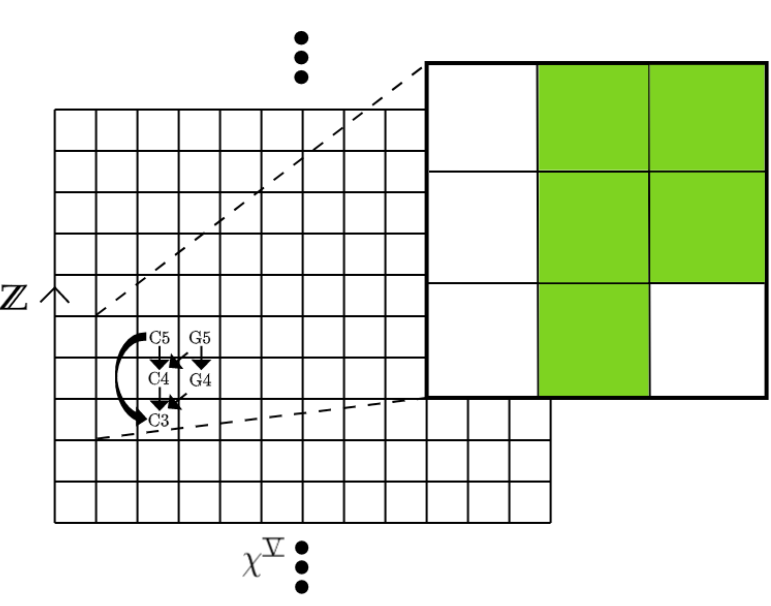}
    \caption{Visualisation of the final grid structure, where green cells represent a Boolean value of true.}
    \label{fig:gridfinal}
\end{figure}

Further, one can imagine that trying to represent a large number of nodes (and therefore a likely larger number of edges) renders this representation visually messy, and hard to follow or decipher. By placing each node into a grid (with the horizontal axis representing the pitch chromas ordered according to the circle of fifths, and the vertical axis representing the pitch height), and restricting the edges to only the first three harmonics, this is alleviated. From here, it becomes clear that it is in fact possible to dispose of the notion of this representation of a graph altogether: by removing the edges (the information from which becomes implicit), and instead assigning a Boolean value to each cell, representing whether the note is audible or not (Figure~\ref{fig:gridfinal}). 
In order to apply any proposed algorithms, it is necessary to convert the audio signal into the proposed format. Doing so effectively will require reasonably extensive preprocessing (the full extent of which is beyond the scope of this paper), which will need to include,
\begin{itemize}
    \item taking the signal from the time-domain to the frequency domain, and discretising the frequency-domain signal (by binning the continuous signal into the semitones of the western musical scale), likely through a constant-Q transform \citep{brown1992efficient}, or other Non-Stationary Gabor Transform (NSGT);
    \item eliminating the inharmonic part of the signal (in any domain). NB: the more effectively this noise is reduced, the better the algorithms will perform (as is true of most approaches); and
    \item translating the resulting signal into the proposed grid format (Algorithm~\ref{alg:interpalg}). Computationally, this can be efficiently represented as a sparse matrix of 0s and 1s,
\end{itemize}
in some order.

The problem of pitch detection is then reduced to the problem of finding the decomposition of the grid (into shapes, or configurations) that corresponds to the notes played in the input signal. As becomes apparent, this involves discarding a number of false positive cases from the interpretation. From a graphical perspective, this is somewhat equivalent to identifying the nodes that correspond to fundamentals, removing them, and repeating the process until no more are present.

\begin{figure}
    \centering
        \includegraphics[scale=0.55]{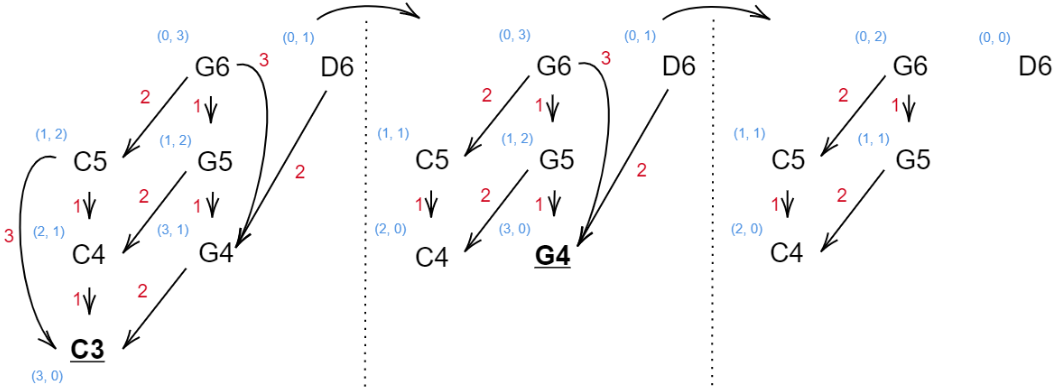}
    \caption{(\textit{left}) A directed graph depicting C3 and G4 (and each of their first three harmonics sounding). The degrees of each node are shown in blue. The following steps represent the steps of the simple algorithm described. The bolded/underlined note at each step is the one selected as a fundamental.}
    \label{fig:graph_alg}
\end{figure}

For example, consider the graph in which (C, 3) and (G, 4) are sounding along with their first three harmonics (Figure~\ref{fig:graph_alg}). By annotating each node, $\nu$, with its indegree ($\text{deg}^{-}(\nu)$) and outdegree ($\text{deg}^{+}(\nu)$), some nodes present as sinks (i.e. with degree $(n^-, 0)$ for some $n^->0$), and some as sources (i.e. with degree $(0,n^+)$ for some $n^+>0$). For the purposes of pitch estimation, and because of the chosen edge direction (from harmonic to fundamental), a sink with indegree three is always a fundamental with its first three harmonics present. Thus, a simple (yet somewhat effective algorithm) is to take each node with indegree three (for each distinct part of the graph, as it may not be connected), categorise them as fundamentals, and remove all categorised nodes. This can then be iteratively applied to the graph until no sinks with indegree three remain (as shown in Figure~\ref{fig:graph_alg}). Clearly this algorithm is a vast oversimplification of the problem, but it nicely illustrates the benefits of visual approaches.

The following sections will build upon this grid-based model, proving some useful properties about it, and presenting some algorithms that utilise them.

\begin{figure}
    \centering
    \includegraphics[scale=0.7]{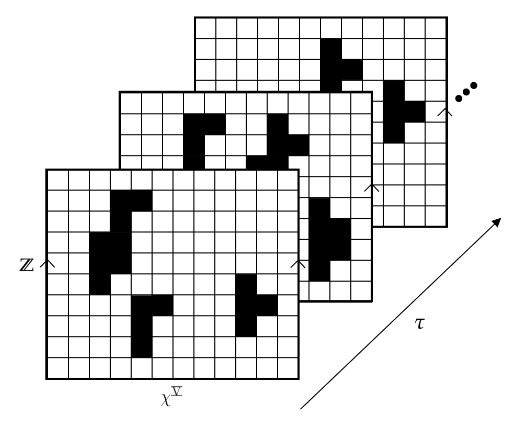}
    \caption{Visualisation of a sequence of interpretations (representing temporal slices), indexed by $\tau$.}
    \label{fig:time_grids}
\end{figure}

\section{The Proposed Model}\label{sec:proposed}
As in Section~\ref{sec:reaching}, let $\chromas$ be the set of pitch chromas, ordered by fifths - $\{C,G,D,A,...,F\}$. Let the grid, $\notes \coloneqq \chromas \times \mathbb{Z}$, where any note in the western harmonic scale is a point in $\notes$, i.e  a pitch chroma/pitch height pair $\nu_{i,j} \in \notes$. That is, $\nu_{i,j}$ represents pitch chroma indexed by $i\in\chromas$ (i.e. the circle of fifths), with pitch height $j \in \mathbb{Z}$. $\notes$ forms the backbone of the model. As the circle of fifths exhibits a periodic nature, the left and right edges of the grid may be identified, or \textit{glued}, to realise $\notes$ as a discretised infinite cylinder. Further, define the predicate, $\interp_{\tau}\colon\notes\to\mathbb{B}$,
\begin{align}
\mathcal{I}_{\tau}(\notese)  =  \left\{ \begin{array}{cc} 
                \top & \hspace{5mm} \text{ if  } \notese \text{ is observed at $\tau$} \\
                \bot & \hspace{5mm} \text{ if  } \notese \text{ is not observed at $\tau$. }  \\
                \end{array} \right.
\end{align}
This is called an interpretation (i.e. of $\notes$), and can be seen as a single time slice of a signal, indexed by an instantaneous point in time, $\tau$. By viewing a musical signal as a sequence of temporal slices, an interpretation of the signal for any given $\tau$ by the pair $(\notes, \mathcal{I}_{\tau})$ is obtained (Figure~\ref{fig:time_grids}).

By viewing the ordered collection of pairs as a whole, one can uniformly stretch each slice, and identify the appropriate $\notes$ faces to create a three-dimensional heatmap, with each note now represented by a cube as opposed to a square. Thus, the interpretations are now indexed by an interval, where previously they had been indexed by an instantaneous point in time, that is, 

\begin{equation}
    (\notes, \mathcal{I}_{\tau}) \longmapsto (\notes, \mathcal{I}_{[\tau, \tau + \text{$1$})}).
\end{equation}

In general, when referring to \textit{any} interpretation henceforth, $\mathcal{I}_{\tau}$ may be replaced by $\mathcal{I}$ for simplicity.

By projecting onto hyperplanes parallel to the faces of the cuboid, one can consider the signal from different perspectives - that is, with constant time, constant pitch chroma, or with constant pitch height. For example, considering the projection with constant pitch height, one can elicit a pitch contour representation of the signal. Further, by viewing the heatmap as a translucent construction, it is possible to consider all aspects simultaneously. 

Let $f_i$ denote the $i^{\text{th}}$ harmonic, with $f_0$ being the corresponding fundamental. Depending on the pitch chroma, the 2$^{\text{nd}}$ harmonic, $f_{2}$, may or may not cross the octave boundary (i.e. be in the next octave up from the 1$^{\text{st}}$ harmonic). For example, the first three harmonics of $(C, n)$ are $(C, n + 1)$, $(G, n + 1)$, and $(C, n + 2)$, whereas the first three harmonics of $(F, n)$ are $(F, n + 1)$, $(C, n + 2)$, and $(F, n + 2)$. By considering each chroma in $\chromas$, it is clear that the presence of $\{f_0, f_1, f_2, f_3  \} \subset \notes$ make up one of two shapes; a turnstile shape, $\vdash$, or a gamma shape, $\Gamma$, depending on the position of the fundamental. Denote the set $$ \chi_{\vdash} = \{C, C\sharp, D, E\flat, E \}, \quad \text{with} \quad   \chi_{\Gamma}  =  \chromas \setminus \chi_{\vdash} \quad \text{as its complement, }$$
 and let $\pi_{\chi}$, $\pi_{y}$ be the projection of $\notes$ onto the horizontal and vertical axes respectively. Then, when $\pi_{\chi}(f_0) \in \chi_{\vdash}$ one observes the $\vdash$ shape, and $\Gamma$ otherwise.
 
 This is shown on Figure \ref{fig:gamma_and_turnstile}, where a fundamental is denoted by $\bullet$ and its harmonics by $\times$. 

\begin{figure}
    \centering
    \includegraphics[scale=0.7]{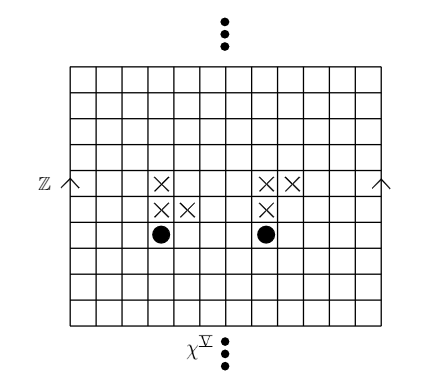}
    \caption{Demonstration of the $\vdash$ and $\Gamma$ shapes in $\notes$.}
    \label{fig:gamma_and_turnstile}
\end{figure}

This model (particularly the use of $\chromas$ as opposed to a chromatically-ordered column set) is chosen such that the pattern exhibited by a fundamental and its first three harmonics (i.e. $\vdash/\Gamma$) appears spatially-compact. This serves to make these patterns more easily discernible, and is also of use when looking to decompose more complicated polyphonic signals into their constituent parts.

The different cells, or notes, on the cylinder can be related to each other by considering a group action on $ \notes$. Let $\delta$ and $\omega$ denote the generators of $\mathbb{Z}_{12}$ (the integers modulo 12) and $\mathbb{Z}$, respectively. Then define a group action $\mathbb{Z}_{12} \times \mathbb{Z} \circlearrowleft \notes$ as follows. $\delta$ and $\omega$ induce maps on $\notes$ by

\begin{align*}
  (\delta, \mathbbm{1}_{\mathbb{Z}}): \  & \notes \rightarrow \notes \qquad \ {(\mathbbm{1}_{\mathbb{Z}_{12}},\omega): \notes \rightarrow \notes\ \ \ }\\
&  \notese \mapsto \nu_{i+1,j},  \qquad \qquad  \qquad   \notese \mapsto \nu_{i,j+1},
\end{align*}
where $\mathbbm{1}_{\mathbb{Z}}$ and $\mathbbm{1}_{\mathbb{Z}_{12}}$ are the identity elements in $\mathbb{Z}$ and $\mathbb{Z}_{12}$, respectively. In other words, $(\delta, \mathbbm{1}_{\mathbb{Z}})$ acts on the cylinder by rotating it clockwise by one cell, while applying $(\mathbbm{1}_{\mathbb{Z}_{12}}, \omega)$ corresponds to  a vertical shift of one cell downwards. Hence a map
\begin{equation}
    \mathbb{Z}_{12} \times \mathbb{Z} \times \notes \rightarrow \notes : \ (\delta^k, \omega^l, \notese) \mapsto \nu_{i+k, j+l},
\end{equation}
is achieved, where $k,l \in \mathbb{Z}$, and $\delta^{12 n}$ for any integer $n$ is the identity. For notational simplicity $(\delta, \mathbbm{1}_{\mathbb{Z}})$  is identified with $\delta$, and similarly for $\omega$. Note that this means that,  relative to a reference point, $\delta$ translates the note by a fifth, and $\omega$ moves the note up an octave. 

In terms of this action,
\begin{equation*}
    \omega(f_0) = f_1, \qquad  \omega \delta (f_0) = f_2,\qquad \text{and} \qquad \omega^2 (f_0) = f_3,
\end{equation*}
for $\pi_{\chi}(f_0) \in  \chi_{\vdash}$, where $\omega \circ \delta$ is identified with $\omega \delta$.  For $\pi_{\chi}(f_0) \in \chi_{\Gamma}$  the above holds with the exception of $f_2$ which in this case is given by
\begin{equation*}
    \omega^2 \delta (f_0) = f_2.
\end{equation*}

Furthermore, using this action, the $\vdash$ and $\Gamma$ shapes may be written as,
\begin{equation}\vdash \, = \{\mathbbm{1},\omega,\omega\delta,\omega^{2}  \}, \qquad \Gamma = \{\mathbbm{1},\omega,\omega^2\delta,\omega^{2}  \},
\end{equation}
where it is understood that by applying all elements of $\vdash$ to a note traces out the turnstile shape, and similarly for $\Gamma$. In other words, considering the $\vdash$ case,  for each fundamental which is mapped to $\top$ by $\mathcal{I}$, there exist harmonics $\omega (f_{0})$,  $\omega\delta (f_{0})$, and $\omega^{2}(f_{0})$ such that each of these are also mapped to $\top$ by $\mathcal{I}$,
\begin{equation} 
    \forall_{\nu\in\notes}\big[(\mathcal{F}(\nu))\to(\mathcal{I}(\omega(\nu))\land\mathcal{I}(\omega\delta(\nu))\land\mathcal{I}(\omega^{2}(\nu)))\big],
\end{equation}
given that $f_{1}$, $f_{2}$, and $f_{3}$ are observed (audible). Here $\mathcal{F}(\nu)$ is a predicate that returns $\top$ iff $\nu$ is a fundamental. Of course, such a construct does not exist in practice, but in essence, the end result of a perfect pitch estimation algorithm is this function, such that it best describes the ground truth of the signal. As before, the $\Gamma$ case is equivalent under the replacement $\omega \delta \mapsto \omega^2 \delta$. 

By observation of the corresponding $\vdash$ and $\Gamma$ shapes over the circle of fifths, it is noted that three two-shape configurations exist - namely $\Gamma\Gamma$, $\Gamma\vdash$, and $\vdash\Gamma$. These are used to categorise a number of properties of the model. 

\begin{definition}[Configuration]
A configuration (denoted as $\Gamma\Gamma$, $\Gamma\vdash$, or $\vdash\Gamma$) represents the shapes generated by fundamentals residing in two adjacent columns in $\notes$. When referring to a false fundamental, $\otimes$, the second column of the configuration always corresponds to that in which $\otimes$ lies. Thus, the first column corresponds to the preceding one - i.e. $\pi_{\chi}(\delta^{-1}(\otimes))$.
\end{definition}

Whilst every fundamental, together with its first three harmonics, exhibit one of the two aforementioned shapes (i.e $\Gamma$ or $\vdash$), the inverse statement is not true. Namely, the presence of a $\vdash$ or $\Gamma$ shape does not imply that the note concerned is a fundamental, as shown in Figure \ref{fig:counterexample}. Here, $\otimes$ denotes a harmonic which presents as a fundamental. Such harmonics are called false fundamentals.

\begin{figure}
    \centering
    \includegraphics[scale=0.45]{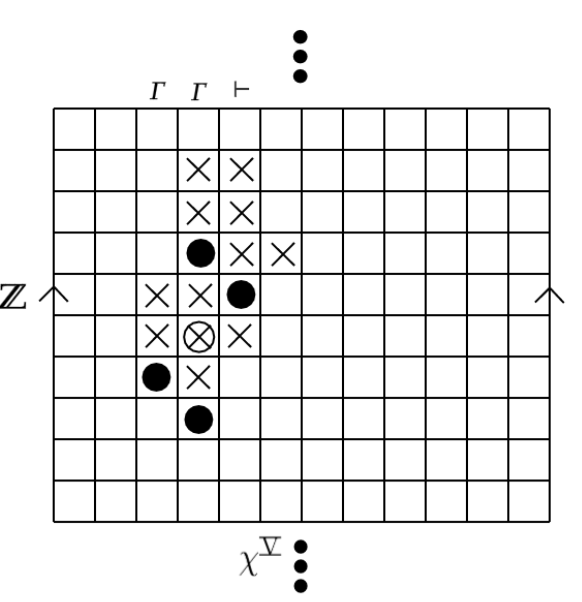}
    \caption{Counterexample showing that not all $\vdash$ shape-exhibiting notes are fundamentals.}
    \label{fig:counterexample}
\end{figure}

Similarly to how a fundamental and its first three harmonics corresponds to either a $\vdash$ or a $\Gamma$, a given harmonic could only have arisen from a fundamental related to it by the \textit{inverse} shape, i.e either  \reflectbox{$\vdash$} or  \reflectbox{L}.

The inverse shapes may be written as, 
\begin{subequations}
\begin{equation}
    \text{\reflectbox{$\vdash$}} = \{\sigma^{-1} |\ \sigma \in\ \vdash  \} = \{\mathbbm{1}, \omega^{-1}, (\omega \delta)^{-1}, \omega^{-2}  \},
\end{equation}
\begin{equation}
    \text{\reflectbox{L}} = \{\sigma^{-1} |\ \sigma \in \Gamma  \} = \{\mathbbm{1}, \omega^{-1}, (\omega^2 \delta)^{-1}, \omega^{-2}  \},
\end{equation}
\end{subequations}

where $\sigma^{-1}$ is the inverse of $\sigma$ with respect to the group structure.

Additionally, define the function $\Psi(\chi_i)$ for any $\chi_i \in \chromas$ as,
\begin{align}
\Psi(\chi_i) =\left\{ \begin{array}{cc} 
                \vdash & \hspace{5mm} \text{ if } \chi_i \in \chi_{\vdash} \\
                \Gamma & \hspace{5mm} \text{ otherwise} . \\
                \end{array} \right.
\end{align}

Intuitively this function takes a given chroma (i.e. $\pi_{\chi}(\nu)\in\chromas$), and returns the set of group elements that trace out the corresponding shape when applied to a note with this chroma. 

The generator of a note is defined as the fundamental that deposited the corresponding frequency. Note that the generators of a note may sit both in the same, or proceeding column to itself. This means then when enumerating the possible generators in most cases (i.e. not $\Gamma\Gamma$), it is necessary to consider notes in \reflectbox{$\vdash$} $\bigcup$ \reflectbox{L}. In many cases there may be multiple generators for a single note. 

\begin{figure}
    \centering
    \includegraphics[scale=0.45]{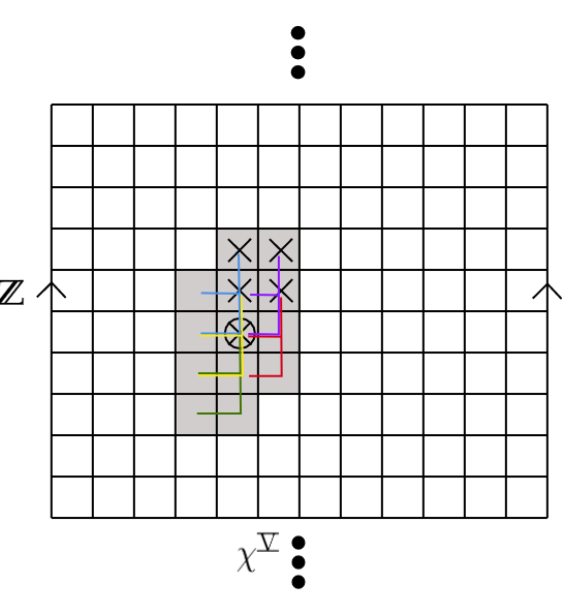}
    \caption{Figure showing where a false fundamental and each of its apparent harmonics could have been generated from, with colours tracing out \reflectbox{$\vdash$} and \reflectbox{L} for each note (overlaying both $\Gamma$ and $\vdash$ situations).}
    \label{fig:neighbourhood}
\end{figure}

Suppose an interpretation is is given such that a false fundamental is present. By investigating the possible positions for this note in $\notes$, there is a finite region containing the fundamentals that could have created the false fundamental and its first three apparent harmonics. In any of the possible positions for the false fundamental, this region is contained in the region given by $\varhexagon (\otimes)$, with $\varhexagon \coloneqq \{\sigma' \sigma|\sigma' \in  \text{\reflectbox{$\vdash$}}
 \bigcup \text{\reflectbox{L}}, \sigma \in \vdash \bigcup \Gamma \}$, as shown in Figure \ref{fig:neighbourhood}. This holds by construction. 
 
 Consequently, in order to check whether a note could be generated by some other note, it is sufficient to search a $3\times5$ area centred on the note. Though this proves sufficient from the perspective of generators, there are a number of edge cases (i.e. false fundamentals) that will naively result in false positives. 

Hence, the problem of multi-pitch estimation is reduced to that of distinguishing between fundamentals ($\bullet$) and harmonics masquerading as fundamentals ($\otimes$).
 
 \section{Fantastic Edge Cases (and Where to Find Them)}\label{sec:edge}
This section considers only cases in which a single false fundamental ($\otimes$) is present. In reality, this is expected to be enough of a generalisation as long as algorithms consider false fundamentals sequentially, such that $\notes$ is traversed along 
\begin{equation}
    (\delta^{i}(\omega\delta^{12})^{j-1})(1,1),\ \forall_{i\in\{1,\cdots,12\}}\forall_{j\in\{1,\cdots,10\}},
\end{equation}
that is, left-to-right, bottom-to-top.

\begin{definition}[Edge Case]

An edge case is a set of notes (i.e. fundamentals and their first three harmonics), in which a note that presents as a fundamental, is in fact not one - as seen in Figure~\ref{fig:counterexample}. That is,
$$\otimes(\nu)\ =\ \big(\forall_{\sigma\in\Psi(\nu)}\big[\mathcal{I}(\sigma\nu)\big]\big)\land \neg\mathcal{F}(\nu).$$

\end{definition}

By considering \reflectbox{$\vdash$} $\bigcup$ \reflectbox{L} for each constituent note (similar to Figure~\ref{fig:neighbourhood}), it is possible to construct logical expressions for the generators of any edge case. Through knowledge of the specific configuration (which is always known for a given note), it is possible to use the appropriate subset of \reflectbox{$\vdash$} $\bigcup$ \reflectbox{L}. The logical expressions are as follows,

\begin{subequations}
\begin{equation*}
    f_0 = \omega^{-1} \lor \, \omega^{-2} \lor \left\{
\begin{array}{ll}
      (\omega\delta)^{-1}  & \text{ if } \Psi(\pi_{\chi}(f_0)-1)\ =\ \vdash \\ 
      \omega^{-2}\delta^{-1} & \text{ otherwise} , \\
\end{array} 
\right. 
\end{equation*}
\begin{equation*}
    f_1 = \omega \lor \omega^{-1} \lor \left\{
\begin{array}{ll}
     \delta^{-1}  & \text{ if } \Psi(\pi_{\chi}(f_0)-1)\ =\ \vdash \\
      (\omega\delta)^{-1} & \text{ otherwise} , \\
\end{array} 
\right.
\end{equation*}
\begin{equation*}
    f_2 = \omega\delta \lor \, \delta \lor \left\{
\begin{array}{ll}
      \omega^{-1}\delta & \text{ if } \Psi(\pi_{\chi}(f_2))\ =\ \vdash \\ 
      \omega^{2}\delta & \text{otherwise}, \\
\end{array} 
\right.
\end{equation*}
\begin{equation*}
    f_3 = \omega \lor \, \omega^{2} \lor \left\{
\begin{array}{ll}
      \omega \delta^{-1}  & \text{ if } \Psi(\pi_{\chi}(f_0)-1)\ =\ \vdash \\ 
      \delta^{-1} & \text{ otherwise}. \\
\end{array} 
\right.
\end{equation*}
\end{subequations}
Note that these actions are all relative to the fundamental, and the interpretation functions are omitted for brevity - that is, $\mathcal{I}_{\tau}(\sigma)$ $\mapsto$ $\sigma$ for all $\sigma = \omega^{i}\delta^{j}$, $(i, j)\in\mathbb{Z}\times\mathbb{Z}$.

Further, it is possible to define basic edge cases;

\begin{definition}[Basic Edge Case]
 A basic edge case is an edge case such that each constituent note has at most one generator.
\end{definition}
More specifically, this means that each apparent harmonic (including the false fundamental itself) has \textit{precisely} one generator (as if any part had zero generators it wouldn't be a false fundamental). This can be trivially represented by replacing the logical ors ($\lor$) with xors ($\oplus$) - for example for $f_0$,

\begin{equation*}
    f_0 = \omega^{-1} \oplus \, \omega^{-2} \oplus \left\{
\begin{array}{ll}
      (\omega\delta)^{-1}  & \text{ if } \Psi(\pi_{\chi}(f_0)-1)\ =\ \vdash \\ 
      \omega^{-2}\delta^{-1} & \text{ otherwise}. \\
\end{array} 
\right. 
\end{equation*}

The same replacement holds for $f_1$, $f_2$, and $f_3$.

Following from this, it is possible to enumerate every basic edge case for a specific configuration, and ascertain the total number. Initially the answer for four choices, each with three options would simply be $3^4 = 81$. Due to overlap in which generators satisfy the constituent, however, the actual result is significantly lower, and can be enumerated with a simple counting method (Table~\ref{tab:counting}). Note that $f_2$ has been omitted as it has no overlap (and therefore, multiplying the end result by 3 is sufficient).

\begin{table}[]
\centering
\begin{tabular}{@{}lll@{}}
\toprule
\textbf{$f_0$}        & \textbf{$f_1$} & \textbf{$f_3$}      \\ \midrule
$\omega^{-1}$         & -              & $\omega^2$          \\
                      &                & $\omega\delta^{-1}$ \\
$\omega^{-2}$         & $\omega$       & -                   \\
                      & $\delta^{-1}$  & $\omega^2$          \\
                      &                & $\omega\delta^{-1}$ \\
$(\omega\delta)^{-1}$ & $\omega$       & -                   \\
                      & $\delta^{-1}$  & $\omega^2$          \\
                      &                & $\omega\delta^{-1}$ \\ \midrule
                      & Total          & 8 $\times$ 3  = 24 \\ \bottomrule \\
\end{tabular}
\caption{Table showing the enumeration of possible basic edge cases for the $\vdash\Gamma$ configuration.}
\label{tab:counting}
\end{table}

The same holds true for the $\Gamma\Gamma$ and $\Gamma\vdash$ configurations, as although the composed actions are different, there are still the same overlaps ($f_0/f_1: \omega^{-1},\ \text{and}\ f_1/f_3: \omega$), and the same number of overall choices.

One might be tempted to claim, therefore, that there are $24\times 3 = 72$ basic edge cases. Though technically this may be true, the choice was made to instead define a number of basic edge types - similar to Lent Davis and Maclagan's definition of cap types regarding the card game SET \citep{davis2003card}. This allows for comparisons to be made irrespective of configuration. Further, there are a number of invariants that hold across all configurations, for all basic edge types. 

Let 
\begin{equation}
    g: \notes\to\notes
\end{equation} 

be the map sending a note to its generating fundamental.

Further consider the triple $g(f_0,\ f_1,\ f_3) = (g(f_0),\ g(f_1),\ g(f_3))$ constructed from letting $g$ act on a fundamental and its first and third harmonics. It is not necessary to consider $f_2$ as it has no overlap with the other constituent parts (as shown below, with the overlaps made bold for clarity),
\begin{subequations}
\begin{alignat}{3}
    f_0:\qquad& \bm{\omega^{-1}}&&\oplus\omega^{-2}& &\oplus(\omega\delta)^{-1} \\
    f_1:\qquad& \bm{\omega}&&\oplus\bm{\omega^{-1}}& &\oplus\delta^{-1} \\
    f_2:\qquad& \omega^2\delta&&\oplus\omega\delta& &\oplus\delta \\ 
    f_3:\qquad& \bm{\omega}&&\oplus\omega^2& &\oplus\omega\delta^{-1}.
\end{alignat}
\end{subequations}

Any basic edge case can be associated with such a triple, which corresponds to the generating set of the false fundamental and its first and third apparent harmonics. 

\begin{definition}
(Basic Edge Type)

Two basic edge cases are of the same \textit{type} iff their two corresponding triples are related by
\begin{equation}
    \delta\mapsto\omega^k\delta,\ k\in\{-1,0,1\}.
\end{equation}
\label{def:basictype}
\end{definition}

\begin{remark}
Note that if a basic edge case is obtained from another through $\delta\mapsto\omega^k\delta$, then it is possible to move in the opposite direction using the inverse map $\delta\mapsto\omega^{-k}\delta$.
\end{remark}

For any given type, there will be precisely \textit{three} members - with precisely one being associated to each of the three configurations. The members of each type (by configuration) are related as according to Figure~\ref{fig:commutativity}. 

\begin{figure}
    \centering
    \includegraphics[scale=0.6]{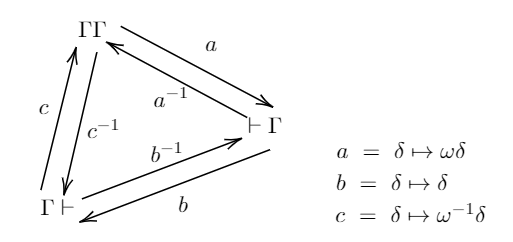}
    \caption{Diagram showing the relationships between members of the same type between different configurations.}
    \label{fig:commutativity}
\end{figure}

This becomes clearer following inspection of Table~\ref{tab:invariants}. Intuitively, the only difference between the $\vdash$ and $\Gamma$ shapes is the shifting of $f_2$. The $f_2$ in a $\Gamma$ shape is obtained from the corresponding $\vdash$ shape by acting $\omega$ on its $f_2$, and its inverse on the contrary.

\begin{lemma}
Being of the same type is an equivalence relation
\end{lemma}
\begin{proof}
For equivalence, the relation must be reflexive, symmetric, and transitive. The reflexive and symmetric properties may be shown by choosing $k=0$, and letting $k \mapsto -k$ in Definition \ref{def:basictype}, respectively. Further, transitivity holds by virtue of the diagram in Figure~\ref{fig:commutativity}.
\end{proof}

\begin{remark}
Note that the types of basic edge case are independent of note configuration. 
\end{remark}

Figure~\ref{fig:basic types} visually represents the eight basic edge types shown in Table~\ref{tab:invariants}.

There are a number of invariants that hold within each type. These are properties which are the same across each edge type, and provide information about their geometric structure. The invariants considered are; back-$\delta$ count ($\backdelta$), edge characteristic ($\epsilon$), and generating structure ($G.S$) which are to be defined in the following. 

\begin{definition}[back-$\delta$ count]

The back-$\delta$ count, $\backdelta$, is a positive integer representing the number of $\delta^{-1}$ occurring in a triple $(g(f_0),\ g(f_1),\ g(f_3))$ associated to to a given edge case.
\end{definition}

\begin{example}
The triple 
$    (g(f_0),\ g(f_1),\ g(f_2)) = (\omega^{-1},\ \omega^{-1},\ \omega\delta^{-1}) 
$
has back-$\delta$ count $\backdelta = 1$.
\end{example}

\begin{definition}[Edge Characteristic]

The edge characteristic, $\epsilon$, is given by the number of distinct fundamentals that are generators for the given note. That is, the number of distinct elements in $(g(f_0),\ g(f_1),\ g(f_3))$.
\end{definition}

\begin{example}
The triple $(g(f_0),\ g(f_1),\ g(f_3)) = (\omega^{-1},\ \omega^{-1},\ \omega\delta^{-1})$ has edge characteristic $\epsilon = 2$.
\end{example}

In addition to $\backdelta$ and $\epsilon$, another invariant is the ``Generating Structure'' (G.S.), which not only considers the number of generating fundamentals, but also which pairs satisfy overlap (i.e. pairs generated by the same note). Naively, there are 5 possible generating structures,
\begin{enumerate}[label=\Roman*.]
\centering
    \item $f_0 = f_1 = f_3$
    \item $f_0 = f_1 \neq f_3$
    \item $f_0 = f_3 \neq f_1$ 
    \item $f_0 \neq f_1 = f_3$ 
    \item $f_0 \neq f_1 \neq f_3$
\end{enumerate}

By construction, cases \Romannum{1} and \Romannum{3} can never occur. Hence, all basic edge cases exhibit a generating structure of either \Romannum{2}, \Romannum{4}, or \Romannum{5}.

These invariants help to distinguish between different basic edge cases, as well as potentially aiding in decomposing non-basic cases into basic ones. The invariants for each class are enumerated in Table~\ref{tab:invariants}. 

\begin{table*}[t]
\centering
\begin{tabular}{@{}cccccc@{}}
\toprule
\multirow{2}{*}{\textbf{Type}} & \multirow{2}{*}{\bm{$|\delta^{-1}|$}} & \multirow{2}{*}{\bm{$\epsilon$}} & \multirow{2}{*}{\textbf{G.S.}} & \multicolumn{2}{c}{\textbf{Cases}}                                                                                     \\
                               &                                           &                                      &                                & \bm{$\Gamma\Gamma / \Gamma\vdash$}                       & \bm{$\vdash\Gamma$}                                 \\ \midrule
\textbf{1}                     & 0                                         & 2                                    & $f_0 = f_1 \neq f_3$           & \multicolumn{2}{c}{$\{\omega^{-1}, \omega^2\}$}                                                           \\
\textbf{2}                     &                                           &                                      & $f_0 \neq f_1 = f_3$           & \multicolumn{2}{c}{$\{\omega^{-2}, \omega\}$}                                                                  \\
\textbf{3}                     & 1                                         & 2                                    & $f_0 = f_1 \neq f_3$           & $\{\omega^{-1}, \delta^{-1}\}$          & $\{\omega^{-1}, \omega\delta^{-1}\}$             \\
\textbf{4}                     &                                           &                                      & $f_0 \neq f_1 = f_3$           & $\{\omega^{-2}\delta^{-1}, \omega \}$                & $\{(\omega\delta)^{-1}, \omega \}$              \\
\textbf{5}                     &                                           & 3                                    & $f_0 \neq f_1 \neq f_3$        & $\{\omega^{-2},(\omega\delta)^{-1},\omega^2\}$               & $\{\omega^{-2},\delta^{-1},\omega^2\}$                  \\
\textbf{6}                     & 2                                         & 3                                    & $f_0 \neq f_1 \neq f_3$        & $\{\omega^{-2},(\omega\delta)^{-1}, \delta^{-1}\}$           & $\{\omega^{-2},\delta^{-1}, \omega\delta^{-1}\}$              \\
\textbf{7}                     &                                           &                                      &                                & $\{\omega^{-2}\delta^{-1},(\omega\delta)^{-1},\omega^2\}$    & $\{(\omega\delta)^{-1},\delta^{-1},\omega^2\}$          \\
\textbf{8}                     & 3                                         & 3                                    & $f_0 \neq f_1 \neq f_3$        & $\{\omega^{-2}\delta^{-1},(\omega\delta)^{-1},\delta^{-1}\}$ & $\{(\omega\delta)^{-1},\delta^{-1},\omega\delta^{-1}\}$ \\ \midrule
\multicolumn{1}{l}{}           & \multicolumn{3}{r}{\textbf{Total number of cases}}                                                                & \multicolumn{2}{l}{\textbf{$8\cdot 3 = 24$}}                                                                            \\ \bottomrule \\
\end{tabular}
\caption{Table showing the different types of basic edge case, together with their invariants, and elements (excluding $f_2$).}
\label{tab:invariants}
\end{table*}

\begin{figure}
    \centering
    \includegraphics[scale=0.6]{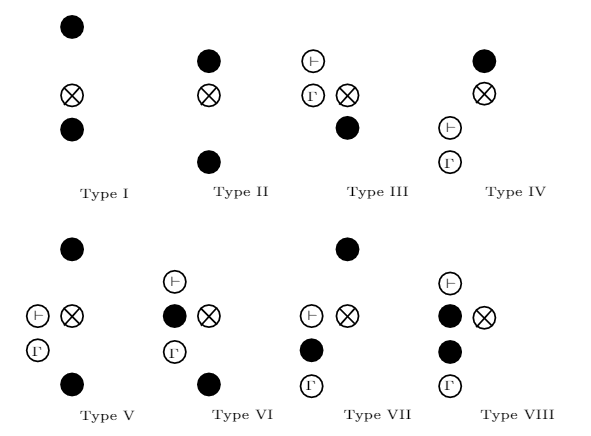}
    \caption{The eight basic edge types represented visually. Each $\bullet$ represents a generator, with $\otimes$ representing the false fundamental, and the unfilled generators containing $\vdash$ or $\Gamma$ denoting the shape drawn out in the $\delta^{-1}$ column (i.e. corresponding to $\Psi(\pi_{\chi}(\delta^{-1}(\otimes)))$.}
    \label{fig:basic types}
\end{figure}

As can be seen in Table~\ref{tab:invariants}, there is only a single case in which the type is not uniquely determined by $\backdelta$, $\epsilon$, and G.S. - namely types 6 and 7. These can be distinguished by considering the position of the generator that sits in the same column as the false fundamental - for type 6, the generator sits below the false fundamental ($\omega^{-2}$), whereas for type 7, the generator sits above the false fundamental ($\omega^2$).

\begin{remark}
Note that, as before, the multiplication by 3 in Table~\ref{tab:invariants} is due to there being no restrictions on the choice of generator for the second harmonic.
\end{remark}

\begin{lemma}
The second ``harmonic" ($f_2$) of a false fundamental ($\otimes$) must be generated from the column directly to the right of the false fundamental (i.e. $\Psi(\pi_{\chi}(\delta(\otimes)))$).
\end{lemma}

\begin{proof}

Assume that there exists some generator $g_i \in \notes$ for $f_2$ that lies in the same column as the false fundamental\footnote{Note that the only other choice would be the column directly to the right of $\otimes$.}. There are two possible cases;
\begin{case}
$\pi_{\chi}(g_{i})\in\chi_{\vdash}$

For $g_i$ to generate the $f_2$ at $\omega\delta(\otimes)$, it would have to lie at $\mathbbm{1}(\otimes)$. Hence, $\otimes$ would no longer be a false fundamental - resulting in a contradiction. $\bot$.
\end{case}
\begin{case}
$\pi_{\chi}(g_{i})\in\chi_{\Gamma}$

The same argument as Case 1, applying the map $\omega\delta\mapsto\omega^2\delta$. Hence, $\bot$.
\end{case}

Thus, as $\bot$ is reached in both possible cases, there can be no such generator for $f_2$. Further, as 
\begin{equation}
    \lnot\exists_{x}.P(x) \equiv \forall_{x}.\lnot P(x),
\end{equation}
all generators for any $f_2$ must lie in the column directly to the right of the false fundamental.
\end{proof}

\begin{proposition}
There are 24 basic edge types.
\label{Prop}
\end{proposition}

\begin{proof}
This follows from Table~\ref{tab:invariants}, and Lemmas 6.1.1 and 6.1.2.
\end{proof}

\begin{corollary}
There are no 5-fundamental basic edge cases
\end{corollary}
\begin{proof}
Following from Table~\ref{tab:invariants}, the largest $\epsilon$ is 3. Including the generator for $f_2$, this leads to a 4-fundamental case. By pigeonhole principle, adding another generator would doubly-satisfy at least one note, meaning that it would no longer be a basic edge case. Hence, there are no 5-fundamental basic edge cases.
\end{proof}

It also proves possible to define restrictions on the presence of generators for a false fundamental (with respect to the chroma configuration in which it sits). In order to do this, the minimum number of generators that must fall in certain proximity (such as the von Neumann and Moore neighbourhoods) to a false fundamental can be considered. 
 
 In terms of the operators $\omega, \delta$, the von Neumann- and Moore-neighbourhoods of some note $\nu$, are the notes generated by acting the elements of the sets
 \begin{equation}
    \{\delta^{\pm 1}, \omega^{\pm 1}\} \quad \text{ and } \quad   \{\delta^{\pm 1}, \omega^{\pm 1}, \omega^{\pm 1} \delta^{\pm 1} \},
\end{equation}
on $\nu$, respectively.
 
 \begin{figure*}[]
     \centering
     \includegraphics[scale=0.6]{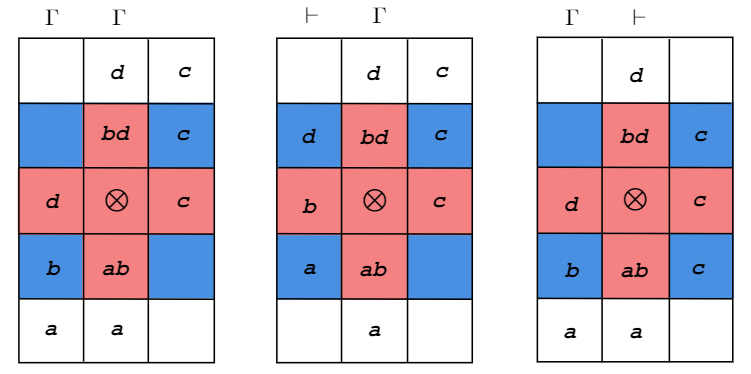}
     \caption{The potential generators (a ($f_0$), b ($f_1$), c ($f_2$), d ($f_3$)) for each configuration, and the von Neumann (red) and Moore (blue) neighbourhoods.}
     \label{fig:neighbourhood_tables}
 \end{figure*}
 
 The problem of choosing basic edge cases with the least generators in these neighbourhoods can be reduced to the problem of choosing some a, b, c, and d (corresponding to generators for $f_0,...,f_3$) from Figure~\ref{fig:neighbourhood_tables}. This can be achieved by choosing generators according to the following order, 
 
 \begin{enumerate}[label=\Roman*.]
    \item Outside both neighbourhoods;
    \item Inside Moore neighbourhood, outside of von Neumann neighbourhood;
    \item Inside von Neumann neighbourhood (and $\therefore$ inside Moore neighbourhood).
\end{enumerate}

If multiple choices are available, it is sufficient to choose any one, without loss of generality, as they have the same effect on the final count as one another. Table~\ref{table:nbh} shows the resulting (minimum) counts of generators in the neighbourhoods for false fundamentals of certain configurations. 

 \begin{table}
 \centering
\begin{tabular}{@{}llll@{}}
\toprule
     & $\Gamma\Gamma$             & $\vdash\Gamma$         & $\Gamma\vdash$             \\ \midrule
a    & $\omega^{-2}$ (I)          & $\omega^{-2}$ (I)      & $\omega^{-2}$ (I)          \\
b    & $(\omega\delta)^{-1}$ (II) & $\delta^{-1}$ (III)    & $(\omega\delta)^{-1}$ (II) \\
c    & $\omega^{2}\delta$ (I)     & $\omega^{2}\delta$ (I) & $\omega\delta$ (II)        \\
d    & $\omega^{2}$ (I)           & $\omega^{2}$ (I)       & $\omega^{2}$ (I)           \\ \midrule
M    & 1                          & 1                      & 2                          \\
v.N. & 0                          & 1                      & 0                          \\ \bottomrule \\
\end{tabular}
\caption{Table showing the minimum generators in the von Neumann (v.N.) and Moore neighbourhoods of a false fundamental given its chroma configuration. \label{table:nbh}}
\end{table}

Though it may seem that choosing a generator that satisfies multiple parts (i.e. $\omega^{-1}$ or $\omega$) may reduce the overall counts, it is always possible in these cases to instead make choices that don't reside in the von Neumann neighbourhood.

Better understanding the occurence of edge cases is an important step towards identifying them in practice, and gives a deeper understanding of the proposed model itself. Sections~\ref{sec:reduction} and~\ref{sec:prevalence} go on to look at reduction of edge cases to potential basic cases, and the experimental prevalence of basic edge types, and Sections~\ref{sec:real} and~\ref{sec:eval} investigate the theoretical basis of the model from a more experimental standpoint.

\section{Reduction and Reducibility of Edge Cases} \label{sec:reduction}
In order to gain a better understanding of the occurrence of edge cases (and therefore the problem of pitch estimation), it proves useful to be able to classify edge cases by which basic edge types they are related to. In order to achieve this, it is necessary to reduce edge cases (i.e. remove redundancy) by removing potential generators such that the false fundamental in question is still preserved. 

\begin{figure}
    \centering
    \includegraphics[scale=0.6]{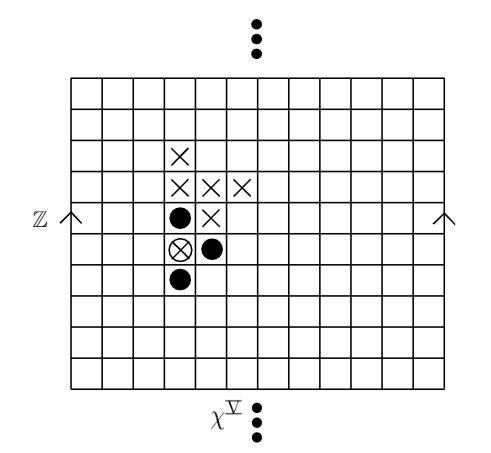}
    \caption{An irreducible non-basic edge case.}
    \label{fig:edge_edge}
\end{figure}

Given a set of generators, $\mathcal{G}$, that lie in $\varhexagon(\otimes)$ for some false fundamental, they are reducible iff any part, $f_n\ \in\ \{f_0,\ f_1,\ f_2,\ f_3\}$ is satisfied more than once (i.e. non-basic) - barring the exception outlined below. Reduction (denoted as $\rightarrow_{g}$) takes $\mathcal{G}$, and removes some given generator $g\ \in\ \mathcal{G}$ such that the false fundamental is still satisfied by $\mathcal{G}\setminus g$,
\begin{equation}
    \mathcal{G}\ \rightarrow_{g}\ \mathcal{G}\setminus g.
\end{equation}
Such a removal of a generator seeks only to remove its `fundamentalness' - it is entirely possible that it could still be generated elsewhere. Indeed, this must be the case for any reduction via a generator in $\vdash \bigcup\ \Gamma$, such as in Figure~\ref{fig:fund_reduction}. Note that reduction is not unique; there may be multiple valid reductions that can be applied to a given set of generators.

\begin{figure}
    \centering
    \includegraphics[scale=0.75]{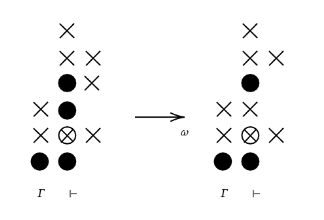}
    \caption{A reduction removing a generator in $\vdash \bigcup\ \Gamma$. Note that $g(f_2)$ is omitted for brevity.}
    \label{fig:fund_reduction}
\end{figure}

It would be reasonable to assume that any non-basic edge case can be reduced, therefore, to one of the eight basic edge cases (Table~\ref{tab:invariants}). On the contrary, however, there exists a case that is both non-basic (i.e. each of its parts are not satisfied exactly once) and irreducible - that is, that no potential generator could be removed whilst preserving the false fundamental (Figure~\ref{fig:edge_edge}). Importantly however, this is the only irreducible non-basic edge case.

\begin{proposition}
The only irreducible non-basic edge cases contain both $\omega^{-1}$ and $\omega$.
\end{proposition}

\begin{proof}
First, it is important to identify the root cause of the irreducibility  of $\{\omega^{-1},\ \omega,\ g(f_2)\}$. Examining the generators of each part of the false fundamental yields,
\begin{align*}
    f_0&:\ \omega^{-1} \\
    f_1&:\ \omega^{-1}\ \lor\ \omega\\
    f_2&:\ g(f_2)\\
    f_3&:\ \omega. \\
\end{align*}
Note that as before, the specific choice of generator for $f_2$ does not matter, as it has no impact on the other generators. 

As $f_1$ is satisfied by multiple generators, it is not a basic edge case. In order to reduce it to a basic edge case, it would be necessary to remove one of the generators of $f_1$ - either $\omega^{-1}$ or $\omega$. Removing either results in $f_0$ or $f_3$ respectively no longer being satisfied, however, and thus neither can be removed whilst still preserving the false fundamental. Thus, there are clearly multiple forms of irreducibility - beyond the basic case. 
\begin{itemize}
    \item Basic irreducibility - each $f_n\ \in\ \{f_0,\ f_1,\ f_2,\ f_3\}$ is satisfied exactly once. Hence, removing any number of generators will result in the false fundamental no longer being satisfied.
    \item Non-basic irreducibility - each $f_n\ \in\ \{f_0,\ f_1,\ f_2,\ f_3\}$ is satisfied, but in such a way that at least one is saturated (i.e. has more than one unique generator), and removing any of the additional generators results in the false fundamental no longer being satisfied. By observation of Figure~\ref{fig:neighbourhood_tables}, this overlap only occurs between $\omega^{-1}$ and $\omega$.
\end{itemize}

It is clear that there exist irreducibile cases containing $\omega^{-1} \oplus \omega$ - i.e. classes I and II (Table~\ref{tab:invariants}). Hence, the only non-basic irreducible edge cases must contain both $\omega^{-1}$ and $\omega$. Further, any case constructed to contain $\omega^{-1}$, $\omega$, and some number of other fundamentals is trivially reducible back to $\omega^{-1}$ and $\omega$. 

Thus, the only non-basically irreducible edge case is the one containing $\omega^{-1}$ and $\omega$. 
\end{proof}
It is worth noting that there are really three such cases, based on arbitrary choice of the generator for $f_2$.

\begin{figure}
    \centering
    \includegraphics[scale=0.6]{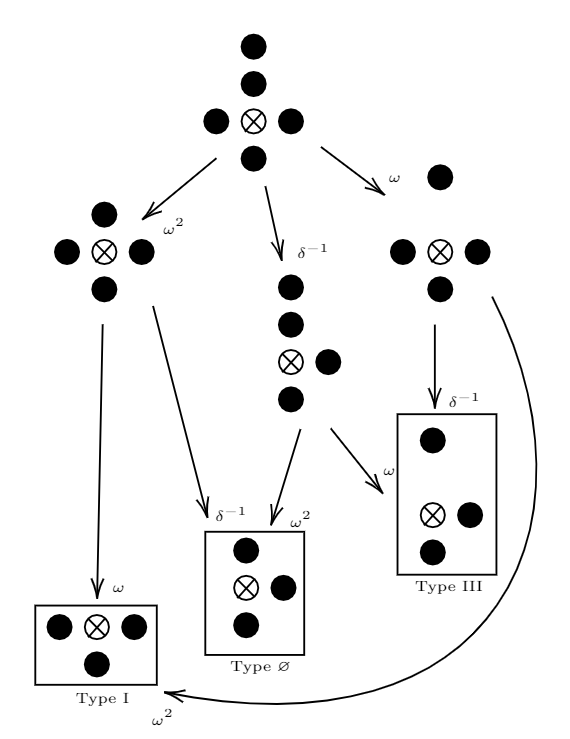}
    \caption{An example of a reduction graph, with each step (arrow) showing a reduction in the set of generators. Note that the special case of $\{\omega,\ \omega^{-1},\ g(f_2)\}$ is denoted as `Type $\varnothing$', and the configuration is $\Gamma\Gamma$ or $\Gamma\vdash$.}
    \label{fig:reduction}
\end{figure}

Through repeated application of all possible reductions to the nodes (to which reduction is yet to be applied), one may obtain a reduction graph for a given set of generators, $\mathcal{G}$ (Figure~\ref{fig:reduction}). Given that no reduction could ever produce a set of generators larger than the input, this graph will additionally be acyclic. Further, the terminating nodes (i.e. $\text{deg}^{-}(\nu) = 0$) of such graphs correspond to irreducible cases. Such a graph can, therefore, be used in order to understand the potential basic edge types that correspond to a given set of generators. 

\section{Prevalence of Basic Edge Types} \label{sec:prevalence}
As previously mentioned, it is important to understand the occurrence of false fundamentals, in order to better differentiate them from genuine ones. One such way is to consider the prevalence of each type in practice. This can be achieved by constructing a reduction graph for sample interpretations (from all combinatorial possibilities) with varying numbers of fundamentals (and their first three harmonics). 

In order to sample interpretations from the total sample space, it proves sufficient to construct them by selecting n unique notes, $\nu_{0}, \nu_{1}, ..., \nu_{n}$ from $\notesa$, treating all such notes as fundamentals, and thus adding them (and their harmonics) to the interpretation. For the charts in Figure~\ref{fig:prevfull}, a sample size of 1000 interpretations was taken for each number of simultaneous fundamentals (0, 120] For each of these interpretations, a naive algorithm (simply classifying all $\vdash$ and $\Gamma$-exhibiting notes as fundamentals) was applied, and a reduction graph derived from each $\otimes$ - where the difference between the input set and the result of the naive algorithm is the set of false fundamentals. From each of these reduction graphs, all terminal nodes were classified either as one of the basic edge types, or the special case, $\varnothing$. In cases where multiple terminal nodes were present, a value was added to each tally such that the sum of all added values was one.

Though 1000 interpretations may at first appear to be a relatively small sample size, it should be noted that this corresponds to 120,000 interpretations sampled on the whole, with an average of 15.75 (16) edge cases per interpretation. These are, as expected, concentrated around the centre of the distribution (of total simultaneous fundamentals), as the number of total possible edge cases peaks around the centre. Thus, on average, each set of 1000 interpretations leads to 1575 edge cases to classify, but with relatively sparse distribution to the tail-ends (i.e. $<10$ simultaneous fundamentals), which resulted in $<100$ edge cases being classified per 1000 interpretations. In order to ascertain a more reliable picture of the makeup of edge cases - particularly with low numbers of simultaneous fundamentals, a significantly larger sample size of 20000 interpretations was used (Figure~\ref{fig:prevlow}). 

Looking at Figure~\ref{fig:prevfull}, it is clear that not all basic edge types are equally common. Subfigure~\ref{fig:prevfull:subfig:pie} gives the overall occurrence of each type, with type 3 being the most common (along with the other three-fundamental cases), and type 5 being the least common (along with the other four-fundamental cases). The special case, $\varnothing$ appears to sit between the two, which intuitively coheres with other observations, as it too is a three-fundamental case - albeit not basic. In general, it is hard to draw meaningful insight from this, which incentivises the use of Subfigure~\ref{fig:prevfull:subfig:stackedge} - looking at the trends of the prevalences as the number of fundamentals changes. 

As Subfigure~\ref{fig:prevfull:subfig:stackedge} shows, at low numbers of simultaneous fundamentals, the three-fundamental cases are significantly more dominant than the four-fundamental cases - constituting almost 100\% of the cases until around 16 simultaneous fundamentals. Beyond this point, the incidence of four-fundamental cases increases significantly - particularly types 6 and 8 - with the special case $\varnothing$ notably occurring increasingly less often. As before, it is hard to directly relate these results to real-world data (i.e. recorded music), which is much more structured than the random samples that were used, but a number of conclusions can still be drawn,
\begin{itemize}
    \item With low numbers of simultaneous fundamentals (eg. string quartet), cases 5-8 are incredibly unlikely to occur.
    \item From Subfigure~\ref{fig:prevfull:subfig:stackall}, it is clear that even at large numbers of fundamentals, the accuracy of even the naive algorithm on polyphonic music - with perfect noise removal, recording, playing, etc. - is above around 75\%.
\end{itemize}
Regarding Subfigure~\ref{fig:prevfull:subfig:decomp}, graphs appear to have between three and five (of a possible nine) terminal nodes, with the average broadly decreasing as the number of fundamentals grows. The trend appears more turbulent towards the left tail, which is likely due to the low number of samples for these numbers of fundamentals. 

By combining this knowledge with a heuristic for the number of fundamentals at a given $\mathcal{I}_{\tau}$, it may be possible to more easily distinguish between fundamentals and false fundamentals by comparing specific examples to the profile laid out above. 

Figure~\ref{fig:prevlow} considers specifically the cases for which there are a low ($<10$) number of simultaneous fundamentals. In these cases, the total occurrence of four-fundamental basic cases is, on average, 3.6\%, with the majority of these weighted towards interpretations with $>6$ simultaneous fundamentals (Subfigures~\ref{fig:prevlow:subfig:pie}~\&~\ref{fig:prevlow:subfig:stackedge}). Particularly interestingly, the most common case in this subset of interpretations is the special case, $\varnothing$, with 20.4\% of the total. Overall, the trend of three-fundamental cases being more common remains, but the ordering within these groupings change - most notably (beyond $\varnothing$'s jump) with type 5 cases being significantly more prevalent than their counterparts compared to the data in Figure~\ref{fig:prevfull}. 

\begin{landscape}
\centering
\begin{figure}[htbp]\label{fig:prevfull}
\vspace*{-1.6cm} 
  \begin{subfigure}[b]{0.5\linewidth}
    \centering
    \includegraphics[scale=0.6]{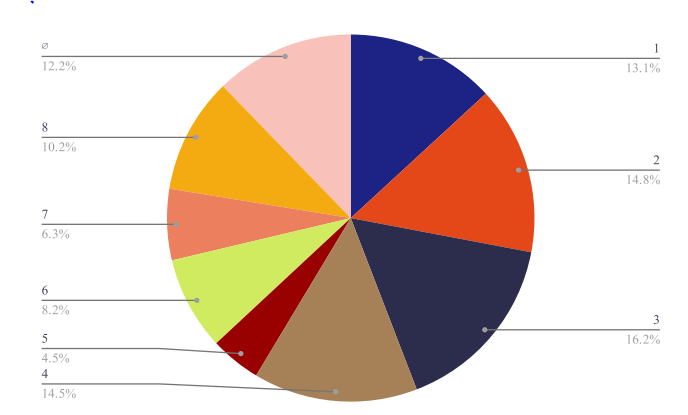}
    \caption{A pie chart showing the average (proportional) prevalence of \\ each edge type, and $\varnothing$.} 
    \label{fig:prevfull:subfig:pie} 
    \vspace{4ex}
  \end{subfigure}
  \begin{subfigure}[b]{0.5\linewidth}
    \centering
    \includegraphics[scale=0.6]{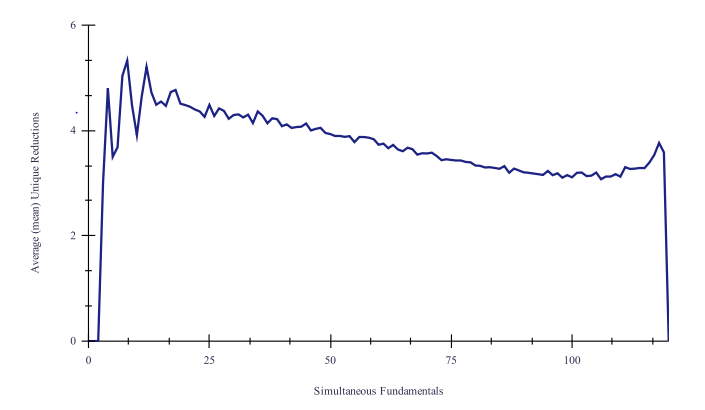}
    \caption{A graph showing the estimated average number of terminal nodes for varying numbers of simultaneous unique fundamentals.} 
    \label{fig:prevfull:subfig:decomp} 
    \vspace{4ex}
  \end{subfigure} 
  \begin{subfigure}[b]{0.5\linewidth}
    \centering
    \includegraphics[scale=0.55]{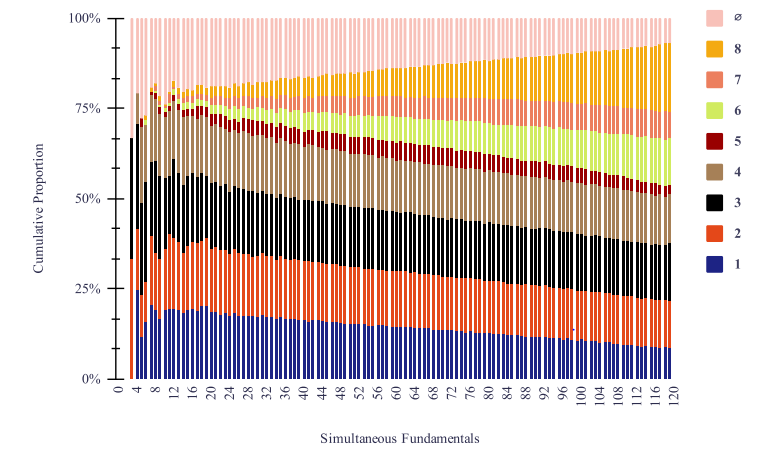}
    \caption{A stacked bar chart showing the change in proportional \\ prevalence of basic edge types and $\varnothing$ as the number of \\ simultaneous unique fundamentals changes.} 
    \label{fig:prevfull:subfig:stackedge} 
  \end{subfigure}
  \begin{subfigure}[b]{0.5\linewidth}
    \centering
    \includegraphics[scale=0.6]{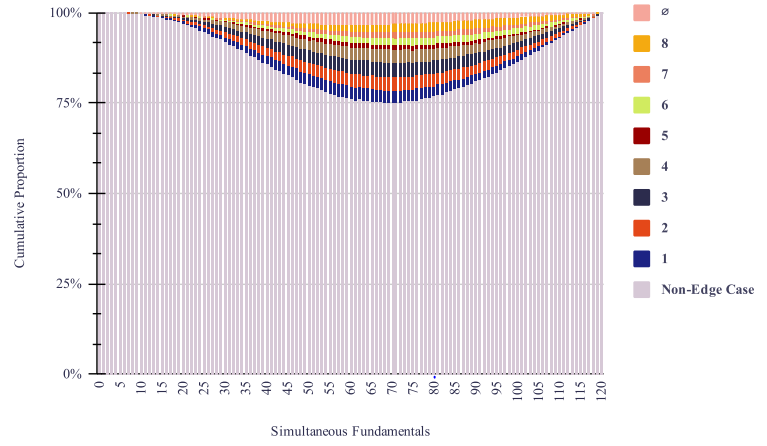}
    \caption{Another stacked bar chart, mirroring subfigure~\ref{fig:prevfull:subfig:stackedge}, but including the proportion of non-edge cases.}
    \label{fig:prevfull:subfig:stackall} 
  \end{subfigure} 
  \caption{Charts depicting various properties relating to the prevalence of basic edge types with respect to the number of simultaneous unique fundamentals.}
  \label{fig:prevfull} 
\end{figure}
\end{landscape}

\begin{landscape}
\begin{figure}[]\label{fig:prevlow}
\vspace*{-1.6cm} 
  \begin{subfigure}[b]{0.5\linewidth}
    \centering
    \includegraphics[scale=0.55]{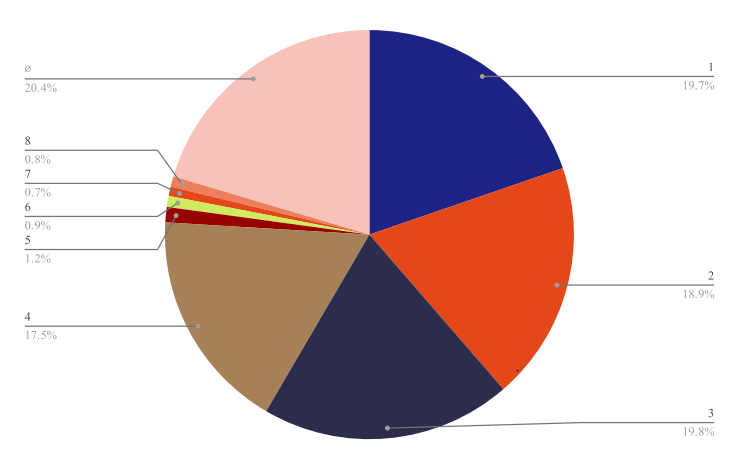}
    \caption{A pie chart showing the average (proportional) prevalence of \\ each edge type, and $\varnothing$ for total simultaneous fundamentals (0, 10].} 
    \label{fig:prevlow:subfig:pie} 
    \vspace{4ex}
  \end{subfigure}
  \begin{subfigure}[b]{0.5\linewidth}
    \centering
    \includegraphics[scale=0.55]{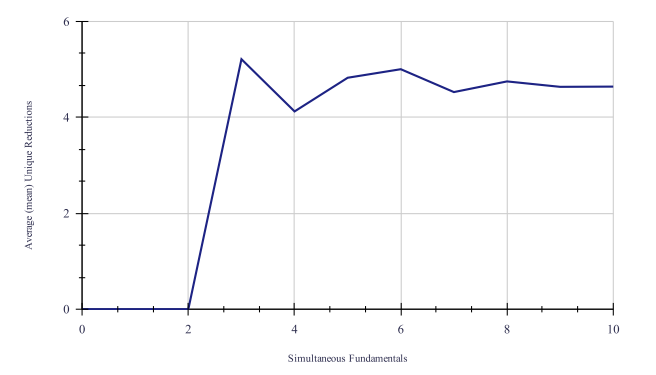}
    \caption{A graph showing the estimated average number of terminal nodes for varying numbers of simultaneous unique fundamentals (0, 10] - corresponding to the left hand side of Subfigure~\ref{fig:prevfull:subfig:decomp}.} 
    \label{fig:prevlow:subfig:decomp} 
    \vspace{4ex}
  \end{subfigure} 
  \begin{subfigure}[b]{\linewidth}
    \centering
    \includegraphics[scale=0.6]{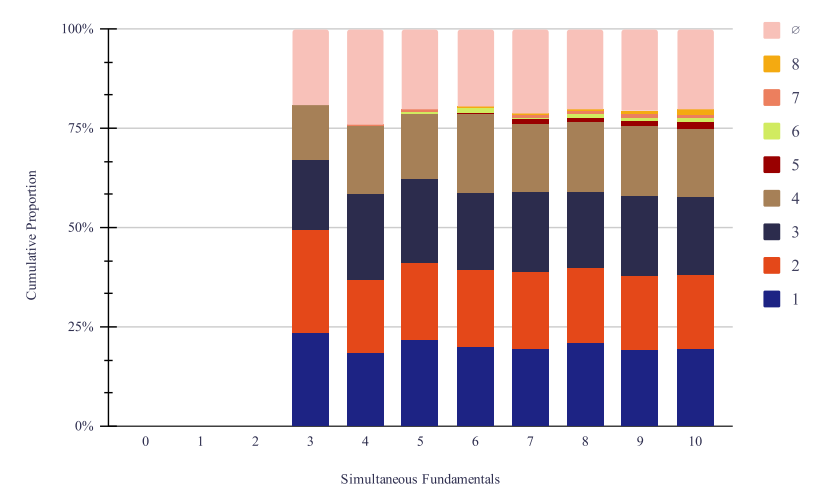}
    \caption{A stacked bar chart showing the change in proportional \\ prevalence of basic edge types and $\varnothing$ for (0, 10] simultaneous unique fundamentals.} 
    \label{fig:prevlow:subfig:stackedge} 
  \end{subfigure}
  \caption{Charts mirroring those in Figure~\ref{fig:prevfull}, but considering only interpretations with (0, 10] simultaneous unique fundamentals.}
  \label{fig:prevlow} 
\end{figure}
\end{landscape}

\section{Experimental Application} \label{sec:real}
Though this model is useful theoretically, in practice, real-world applications are rarely so clear-cut or clean - and will remain so unless there exists some perfect approach to noise removal, amongst other preprocessing. Hence, it is prudent to look not at the discrete, but at the continuous in intepretations, $\mathcal{I}$ - i.e. $\mathcal{I}:\notes\to\mathbb{B}$  becomes $\mathcal{I}:\notes\to\mathbb{R}$.

\begin{figure}
    \centering
    \includegraphics[scale=0.65]{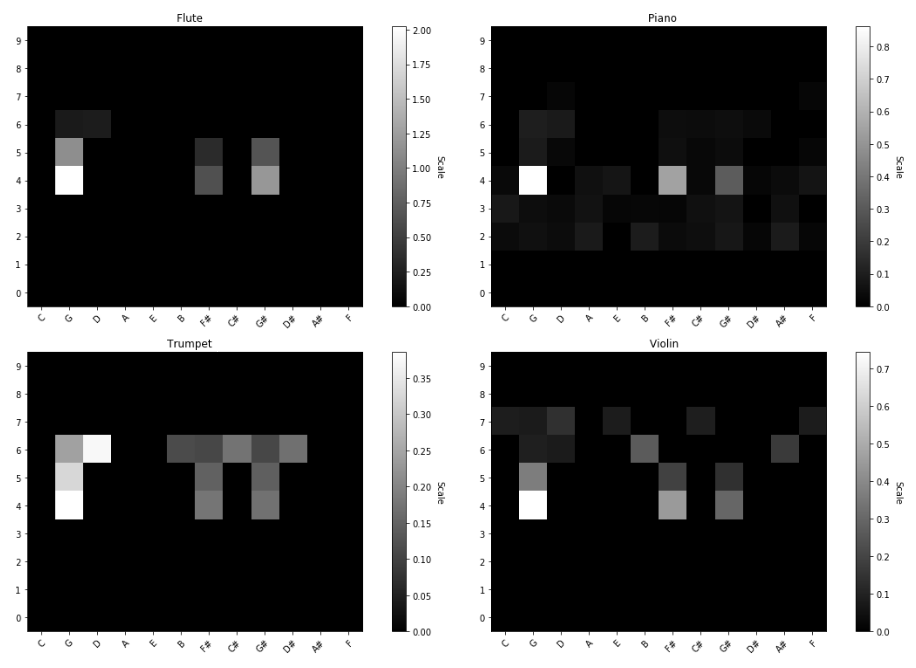}
    \caption{The note G4 being played on a variety of instruments, all exhibiting the $\Gamma$ shape described in Section~\ref{sec:proposed}.}
    \label{fig:iowa_mono}
\end{figure}

Doing so effectively creates a heatmap, in which this additional dimension (perpendicular to $\notes$) represents an `intensity' of each note - for example, their respective amplitudes in the frequency domain \footnote{
This construction can be viewed as a real rank 1 trivial vector bundle, with $\notes$ as a base manifold with trivial topology. In this interpretation, a heatmap is a slice through the topological bundle.}. 
Even in this kind of construction, however, the $\vdash$/$\Gamma$ shapes are very much still prominent - as demonstrated when applied to some monophonic signals from the University of Iowa (Electronic Music Studios) \citep{fritts} (Figure~\ref{fig:iowa_mono}). Here, the intensity is visualised through brightness, with brighter notes representing more prominent frequencies. Though timbrally very different, all of the instruments shown clearly exhibit the $\Gamma$ shape as anticipated.

Despite this, there are clear differences in the prominence of these shapes between the various instruments. Though flute and trumpet exhibit exceptionally clean examples, the clarity in the piano and violin heatmaps is - whilst still interpretable - somewhat diminished. This is likely a result of multiple media (in the case of piano and violin, strings) vibrating in sympathy to the true fundamental - particularly given that the strings are housed in a shared body. Further, the resonance of this body may also have contributed to the noise. 

To build these models, sliding windows were taken from the signal, with a length of 4096 samples, and a hop size of 1024. A constant-Q transform \citep{constantQ} with a Hanning window \citep{hanning} (using the Librosa implementation \citep{librosa}) was then applied to achieve a frequency domain representation binned by the 120 semitones of the western musical scale between C0 and B9 inclusive. These values were then normalised across the signal (not just per window), and plotted as a heatmap using matplotlib \citep{matplotlib}. Further, each window used here corresponds to a unique interpretation, $\mathcal{I}_{\tau}$, where $\tau$ is the start time index of the window.

It is worth noting that the shapes that appear to mirror the fundamentals and their harmonics in chromatically adjacent columns (i.e. F$\sharp$ and G$\sharp$ in the case of Figure~\ref{fig:iowa_mono}) are a result of spectral leakage, which has not been entirely nullified by the Hanning window. In practice, this could likely be removed, or otherwise accounted for in specific algorithms and approaches. 

\begin{figure}
    \centering
    \includegraphics[scale=0.9]{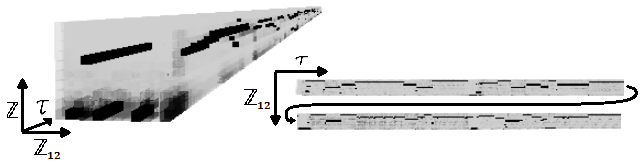}
    \caption{Left: Side-on view of a 3D heatmap of the melody of Bach's ``Ach Gott und Herr", from the Bach10 dataset \citep{duan2015bach10}, with darker colours corresponding to greater amplitudes. Right: Projection of the heatmap onto the $\mathbb{Z}_{12}\times\tau$ plane, eliciting piano roll notation of the piece (albeit ordered by the circle of fifths, and not chromatically).}
    \label{fig:3dplot}
\end{figure}

Looking further, at the three-dimensional heatmap described in Section~\ref{sec:proposed} (with each $\mathcal{I}$ indexed as $\mathcal{I}_{[\tau, \tau+1)}$), Figure~\ref{fig:3dplot} is achieved. With a projection onto the $\mathbb{Z}_{12}\times\tau$ plane, piano-roll notation is achieved as expected. Algorithms working in this space may be able to smooth the estimation in the temporal domain by better-exploiting the temporal aspects of music; it is certainly a great oversimplification to treat each window (and therefore each interpretation) as independent of one another. 

This was achieved using Python's vpython \citep{vpython} module, representing each note as a black (\#ffffff) cube, each with opacity proportional to the loudest note within each specific interpretation\footnote{as opposed to in Figure~\ref{fig:iowa_mono}, where amplitudes are normalised across the whole signal},
\begin{equation}
    Opacity_{\nu} = \frac{|\nu|}{\mathcal{I}_{MAX}} + 0.05,
\end{equation}
with slight linear scaling to make each individual note stand out 
better. Note that $|\nu|$ here refers to the amplitude of a given note, $\nu$, and $\mathcal{I}_{MAX}$ refers to the largest amplitudes in a given interpretation. It is worth noting that an approach using varying shades of grey as opposed to solid black yielded less optimal results.

\begin{figure}
    \centering
    \includegraphics[scale=0.7]{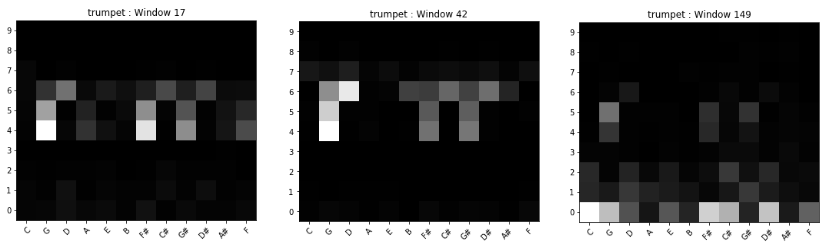}
    \caption{A closer look at how heatmaps differ as the note progresses from onset/attack, to transient, to offset/decay (left-to-right).}
    \label{fig:trumpet_over_time}
\end{figure}

Figure~\ref{fig:trumpet_over_time} shows the notable differences between heatmaps constructed from windows from the onset, transient, and decay of a single note. As expected, the shapes are clearest during the transient, but in general this raises the more profound issue of choosing an appropriate window, or windows, when given chunks of a signal - such as following onset detection. A simple yet effective heuristic is to consider both the total number of bins filled above some threshold $\alpha$ (eg. 3.25$\mu$ \citep{goodman2018real}), and the total magnitude of all bins above this threshold in a given window. That is,
\begin{equation}
    \frac{\sum\limits_{\nu\in N}\nu}{|\{\nu\ |\ \nu\geq\alpha,\ \nu\in N\}|},\ N = \{\nu\ |\ \mathcal{I}(\nu),\nu\in\notesa\};
\end{equation}
effectively the average amplitude of an audible note in the window described by $\mathcal{I}$. Doubtless there are more sophisticated approaches, but this serves its purpose if nothing else but as a benchmark. Should time efficiency not be of particular concern, of course, it may be optimal to consider all windows in a chunk (taking their average result) - only discarding a handful of particularly noisy or otherwise useless ones.  

\begin{figure}
    \centering
    \includegraphics[scale=0.7]{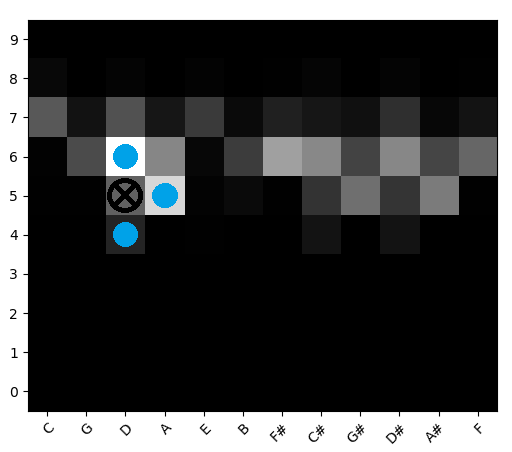}
    \caption{An edge case (specifically $\varnothing$) being exhibited on real data (trumpet) - with D4, A5, and D6 (blue dots) as the fundamentals, and D5 being the false fundamental, $\otimes$.}
    \label{fig:real_edge_case}
\end{figure}

Figure~\ref{fig:real_edge_case} depicts an edge case built up of the notes D4, A5, and D6, all played on trumpet - exhibiting the special case, $\varnothing$. Though masked somewhat by the spectral leakage on the right-hand side of the image, it is clear to see how such edge cases fool the naive algorithm, even when a threshold is utilised to cut out noise. Of course, the use of Algorithm~\ref{alg:interpalg} alleviates this by attempting to remove inharmonic noise, but in doing so, may cause false negatives to arise. Note that when using Algorithm~\ref{alg:interpalg} on real-world data, the same switch from $\mathbb{B}$ to $\mathbb{R}$ applies, insofar as the lines $\mathcal{M}[i, j] = 1$ be replaced by $\mathcal{M}[i, j] = |\nu_{i,j}|$.
That is, the amplitude of the bin as opposed to a boolean representing whether or not the note corresponding to the bin was audible. 

\section{Evaluation} \label{sec:eval}
This section presents a brief evaluation of both a naive algorithm on monophonic music, and a more sophisticated (albeit still simplified) algorithm on theoretical polyphonic samples - similar to those utilised in Section~\ref{sec:prevalence}. As noted beforehand, the intention of this paper (and investigation on the whole) is not to achieve state of the art results on MPE problems, but rather to lay the foundations for more geometrical approaches to them. Thus, the evaluation is brief, but nonetheless provides insight - particularly surrounding future work. 

For monophonic signals, a naive algorithm that treats the relation between fundamentals and $\vdash$/$\Gamma$-exhibiting notes as an equivalence was used. As shown in Section~\ref{sec:proposed}, this is untrue due to the presence of edge cases, but nonetheless when only one fundamental is present, such cases can never occur. In response to the spectral leakage, the algorithm was modified slightly to take not only the shape of the potential fundamental and its harmonics into account, but also the corresponding shapes at $\delta^{\pm5}$.

This was applied to a total of 1395 monophonic samples from the University of Iowa dataset, spanning 17 instruments in total (some of which were categorised into vibrato and non-vibrato playing), resulting in a mean accuracy of 73.74\%. Removing the outliers (violin/viola/cello/double bass (pizz.), and tuba), this average becomes 88.27\%. When considering just whether the pitch chroma is correct (i.e. diregarding octave errors), this increases to 95.08\%. Table~\ref{tab:monoresults} benchmarks this against an implementation of Noll's Harmonic Product Spectrum (HPS) algorithm, also using a Hanning window. 

A table containing a full breakdown of results, broken down by instrument (and vibrato/non-vibrato playing), can be found in Appendix~\ref{app:monotab}.

\begin{table}[]
\centering
\begin{tabular}{@{}llll@{}}
\toprule
                         & \textbf{HPS} & \textbf{Naive} \\ \midrule
\textbf{Overall}         & 58.36\%      & 73.74\%       \\
\textbf{No Outliers}     & 67.73\%      & 88.27\%            \\
\textbf{Chroma Accuracy} & 77.61\%      & 95.08\%               \\ \bottomrule
\end{tabular}
\caption{Table showing the average accuracy of both the naive algorithm and the HPS algorithm as a benchmark when applied to the University of Iowa samples.}
\label{tab:monoresults}
\end{table}

\begin{figure}
    \centering
    \includegraphics[scale=0.17]{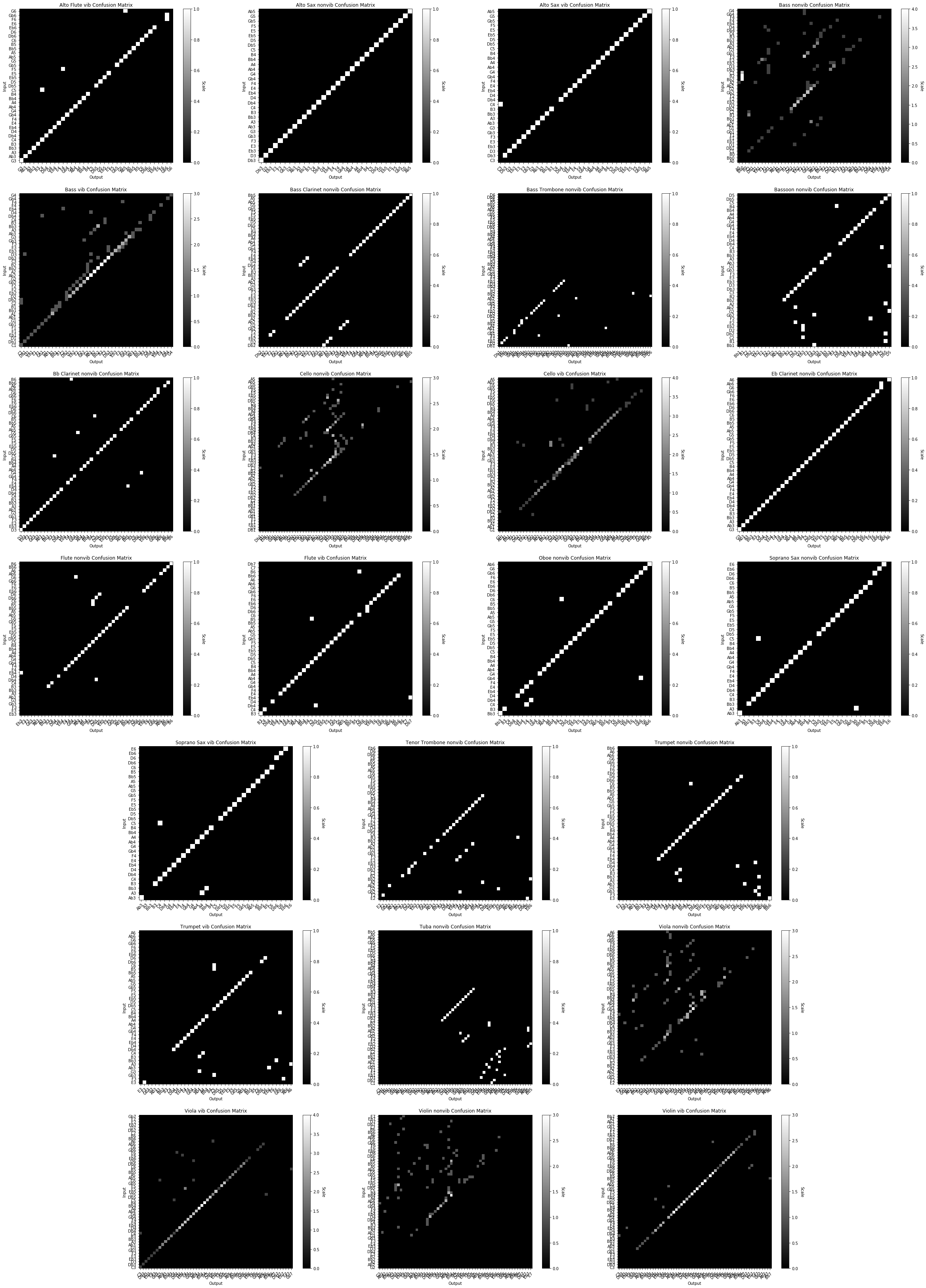}
    \caption{Confusion matrices for each set of samples. From top-left to bottom-right: \textbf{1}) Alto Flute (vib.); \textbf{2}) Alto Sax (non-vib.); \textbf{3}) Alto Sax (vib.); \textbf{4}) Bass (pizz.); \textbf{5}) Bass (arco); \textbf{6}) Bass Clarinet (non-vib.); \textbf{7}) Bass Trombone (non-vib.); \textbf{8}) Bassoon (non-vib.); \textbf{9}) B$\flat$ Clarinet (non-vib.); \textbf{10}) Cello (pizz.); \textbf{11}) Cello (arco); \textbf{12}) E$\flat$ Clarinet (non-vib.); \textbf{13}) Flute (non-vib.); \textbf{14}) Flute (vib.); \textbf{15}) Oboe (non-vib.); \textbf{16}) Soprano Sax (non-vib.); \textbf{17}) Soprano Sax (vib.); \textbf{18}) Tenor Trombone (non-vib.); \textbf{19}) Trumpet (non-vib.); \textbf{20}) Trumpet (vib.); \textbf{21}) Tuba (non-vib.); \textbf{22}) Viola (pizz.); \textbf{23}) Viola (arco); \textbf{24}) Violin (pizz.); \textbf{25}) Violin (arco).}
    \label{fig:confusion_matrices}
\end{figure}

Though, as expected, this approach does not reach state of the art results, it still outperforms HPS by a significant margin. Figure~\ref{fig:confusion_matrices} consists of confusion matrices for each of the instruments (and playing types), showing the algorithm's input (x-axis) against its output (y-axis). Thus, the line $y = x$ is indicative of perfect accuracy, and deviations from this line correspond to errors in the classification. Note that the axes are truncated to match the range of notes tested on each instrument, with the y-axis running chromatically upwards from bottom to top, and the x-axis running chromatically upwards from left to right. Even at first glance, the outliers are relatively clear, and this kind of visualisation has the potential to elicit more profound understanding of how and where an algorithm is failing, and perhaps even (by extrapolation) particular properties of certain instruments that make them more troublesome for pitch detection approaches.

\begin{figure}
    \centering
    \includegraphics[scale=0.55]{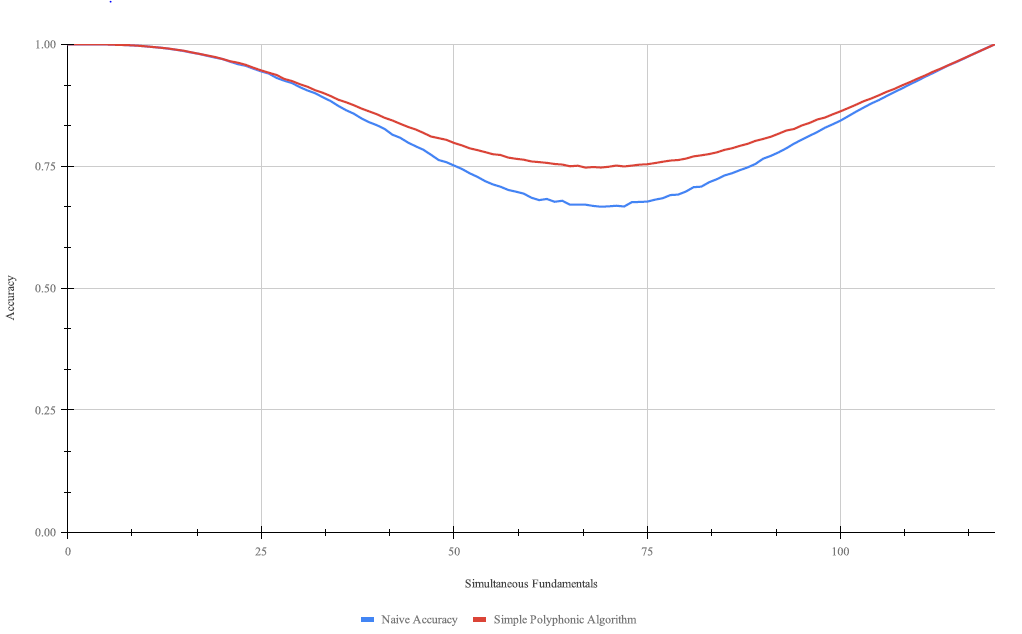}
    \caption{Simulated accuracy for a naive approach (blue), and simple algorithm (red) when applied to sample polyphonic data.}
    \label{fig:polyresults}
\end{figure}

Further, for polyphonic input, a simple extension to the naive monophonic approach, whereby $\notesa$ is traversed left-to-right, bottom-to-top, was utilised. This exploited the assumption that - at least with acoustic music - there will be no undertones. Thus, the bottom-leftmost note with amplitude above some threshold will always be a fundamental \citep{goodman2018real}. The naive algorithm is then applied to subsequent notes, and for each potential fundamental, the possible generators for each of its harmonics are enumerated and checked against the current list of perceived fundamentals (i.e. those that have already been classified as such by the algorithm). If two or more (of a maximum four) of the fundamental and its generators have one or more harmonics that have already been classified as a fundamental, the note is considered to be a false fundamental, and is discarded. The choice of threshold here may seem somewhat arbitrary, but was chosen as there are a significant number of generators that lie above or to the right of the note being classified - most notably the harmonics themselves. Thus, whilst they may themselves be fundamentals, it is unclear at this stage of the algorithm. This choice was then tested empirically, with a value of two (from choices [1, 4]) resulting in the best performance.

As in Section~\ref{sec:prevalence}, 1000 sample interpretations were taken for each number of simultaneous unique fundamentals. The accuracy of this simple approach is benchmarked against the accuracy of the naive approach in Figure~\ref{fig:polyresults}. Though there is a clear increase in accuracy, it is anticipated that it will be possible to build on this simplistic approach using the analysis and techniques outlined in Sections~\ref{sec:edge},~\ref{sec:reduction}, and~\ref{sec:prevalence}, but the implementation of this is beyond the scope of this paper.

\section{Summary and Future Work} \label{sec:future}
Moving forwards, there are a number of improvements and implementations that can be created off of the back of this framework.

Firstly, the simple polyphonic approach can be refined and extended using the characterisations of edge cases outlined in Section~\ref{sec:edge}. This could further be augmented by combining the approaches in Sections~\ref{sec:reduction} and~\ref{sec:prevalence} to perform a `backward pass' over $\notesa$, working right to left, top to bottom, and utilising reduction to reconsider the likelihood of notes presenting as false fundamentals. This can further use the working count of currently perceived generators as a heuristic for the number of distinct simultaneous fundamentals. This kind of forward-backward approach is a fundamental concept in latent variable models - a concept which itself aligns well with the problem of pitch estimation as set out in this paper.

In addition, it may be possible to reformulate the problem as a decomposition of the total heatmap into its constituent (albeit constructively overlapping) $\vdash$ and $\Gamma$ shapes, potentially using spectrogram subtraction to represent this decomposition. This pivots the work towards combinatorics rather than a necessarily more algorithmic perspective. 

Further, one could build on the three-dimensional model to create approaches that effectively utilise the temporal aspects of music in their predictions. One such example may be the extraction of $\vdash$ and $\Gamma$-shaped prisms from the extended heatmap - representing a note, or notes, sounding for some period of time. 

In conclusion, this paper presents a different perspective on approaching pitch estimation problems - doing so from a more geometric perspective. To this end, an idealised model of fundamentals and their harmonics is introduced (and later adapted to real-world scenarios). Importantly, from a geometrical perspective in particular, this model results in spatially-close shapes, namely $\vdash$ and $\Gamma$. Further, it provides a framework on which to approach pitch estimation problems in this way, along with a thorough investigation into the edge cases that occur in this model. 

Though the simple algorithms outlined in Section~\ref{sec:eval} do not provide state of the art results, the intention is instead to provide a solid foundation on which to construct more sophisticated algorithms - particularly utilising the characterisation of, and insight into, edge cases. Furthermore, this is certainly a step towards more intuitive and innovative geometrical solutions in the field of pitch estimation, and Music Information Retrieval on the whole. 

\section*{Acknowledgements}
This work was supported by the EPSRC under Grant: EP/R513167/1.

\newpage
\bibliographystyle{tMAM}
\bibliography{citations}

\newpage
\appendix

\section{Creation of an Interpretation, $\mathcal{I}$}\label{app:algs}
\begin{algorithm}[]
  \caption{Creation of an Interpretation, $\mathcal{I}$, from a Sorted Set of Notes}
  \label{alg:interpalg}
 \begin{algorithmic}
 \renewcommand{\algorithmicrequire}{\textbf{Input:}}
 \renewcommand{\algorithmicensure}{\textbf{Output:}}
 \REQUIRE $\Phi$, a chromatically sorted set of notes
 \ENSURE $\mathcal{M}$, a (matrix) interpretation of the notes in $\Phi$ 
\hrule
\STATE 
\STATE $\mathcal{M}\longleftarrow zeroes(10,\ 12)$
\FOR{$\notese\in\Phi$}
    \STATE // \textit{Is $\notese$ a harmonic of another note?}
    \STATE $f_{1}\longleftarrow\mathcal{M}[i,\ j-1]$ \\
    \STATE $f_{2}\longleftarrow 0$  \\
    \IF{$\Psi (\chi_{i-1})\ =\ \vdash$}
        \STATE $f_{2}\longleftarrow\mathcal{M}[i-1,\ j-1]$ \\
    \ELSE
        \STATE $f_{2}\longleftarrow\mathcal{M}[i-1,\ j-2]$ \\
    \ENDIF
    \STATE $f_{3}\longleftarrow\mathcal{M}[i,\ j-2]$ \\
    \IF{$f_{1}\lor f_{2}\lor f_{3}$} 
    \STATE $\mathcal{M}[i,\ j]\ =\ 1$ \\
    \STATE \textbf{continue} \\
    \ENDIF
    \STATE
    \STATE //\textit{ Is $\notese$ a potential fundamental?}
    \STATE $\phi_{1}\longleftarrow\notese\in\Phi$ \\
    \STATE $\phi_{2}\longleftarrow\bot$ \\ 
    \IF{$\Psi (\chi_{i-1})\ =\ \vdash$} 
        \STATE $\phi_{2}\longleftarrow\nu_{i+1,\ j+1}$ \\ 
    \ELSE
        \STATE $\phi_{2}\longleftarrow\nu_{i+1,\ j+2}$ \\ 
    \ENDIF 
    \STATE $\phi_{3}\longleftarrow\nu_{i,\ j+2}\in\Phi$ \\
    \IF{$\phi_{1}\land\phi_{2}\land\phi_{3}$}
        \STATE $\mathcal{M}[i,\ j]\ =\ 1$ \\
        \STATE \textbf{continue} \\
    \ENDIF
    \STATE 
    \STATE $\Phi\longleftarrow\Phi\setminus\notese$
\ENDFOR 
\STATE \textbf{return} $\mathcal{M}$
 \end{algorithmic}
 \end{algorithm}

\newpage
\section{Full Results - Naive Algorithm (Monophonic)}\label{app:monotab}
\begin{table}[H]
\begin{tabular}{@{}lllrrlrr@{}}
\toprule
\multicolumn{1}{c}{\multirow{2}{*}{\textbf{Instrument}}} & \multirow{2}{*}{\textbf{Type}} & \multicolumn{3}{c}{\textbf{HPS}}                                    & \multicolumn{3}{c}{\textbf{Naive}}                                     \\
\multicolumn{1}{c}{}                                     &                                & 1       & \multicolumn{1}{l}{2}       & \multicolumn{1}{l}{3}       & 1        & \multicolumn{1}{l}{2}        & \multicolumn{1}{l}{3}        \\ \midrule
\textbf{Alto Flute}                                      & Vib                            & 88.89\% & 88.89\%                     & 88.89\%                     & 97.22\%  & 97.22\%                      & 97.22\%                      \\
\multirow{2}{*}{\textbf{Alto Sax}}                       & Nonvib                         & 75.00\% & 75.00\%                     & 81.25\%                     & 100.00\% & 100.00\%                     & 100.00\%                     \\
                                                         & Vib                            & 68.75\% & 68.75\%                     & 75.00\%                     & 100.00\% & 100.00\%                     & 100.00\%                     \\
\multirow{2}{*}{\textbf{Bass}}                           & Pizz                           & 20.19\% & \multicolumn{1}{l}{-}       & \multicolumn{1}{l}{-}       & 53.85\%  & \multicolumn{1}{l}{-}        & \multicolumn{1}{l}{-}        \\
                                                         & Arco                           & 36.54\% & 36.54\%                     & 39.42\%                     & 71.15\%  & 53.85\%                      & 72.12\%                      \\
\textbf{Bass Clarinet}                                   & Nonvib                         & 63.04\% & 63.04\%                     & 65.22\%                     & 100.00\% & 100.00\%                     & 100.00\%                     \\
\textbf{Bass Trombone}                                   & Nonvib                         & 0.00\%  & 0.00\%                      & 29.63\%                     & 44.44\%  & 44.44\%                      & 62.96\%                      \\
\textbf{Bassoon}                                         & Nonvib                         & 45.00\% & 45.00\%                     & 62.50\%                     & 75.00\%  & 75.00\%                      & 95.00\%                      \\
\textbf{Bb Clarinet}                                     & Nonvib                         & 84.78\% & 84.78\%                     & 84.78\%                     & 97.83\%  & 97.83\%                      & 97.83\%                      \\
\multirow{2}{*}{\textbf{Cello}}                          & Pizz                           & 18.00\% & \multicolumn{1}{l}{-}       & \multicolumn{1}{l}{-}       & 46.00\%  & \multicolumn{1}{l}{-}        & \multicolumn{1}{l}{-}        \\
                                                         & Arco                           & 65.26\% & 65.26\%                     & 68.42\%                     & 88.42\%  & 88.42\%                      & 88.42\%                      \\
\textbf{Eb Clarinet}                                     & Nonvib                         & 82.05\% & 82.05\%                     & 82.05\%                     & 94.87\%  & 94.87\%                      & 94.87\%                      \\
\multirow{2}{*}{\textbf{Flute}}                          & Nonvib                         & 94.59\% & 94.59\%                     & 94.59\%                     & 100.00\% & 100.00\%                     & 100.00\%                     \\
                                                         & Vib                            & 94.59\% & 94.59\%                     & 94.59\%                     & 100.00\% & 100.00\%                     & 100.00\%                     \\
\textbf{Oboe}                                            & Nonvib                         & 77.14\% & 77.14\%                     & 97.14\%                     & 100.00\% & 100.00\%                     & 100.00\%                     \\
\multirow{2}{*}{\textbf{Soprano Sax}}                    & Nonvib                         & 84.38\% & 84.38\%                     & 87.50\%                     & 90.63\%  & 90.63\%                      & 90.63\%                      \\
                                                         & Vib                            & 78.13\% & 78.13\%                     & 81.25\%                     & 90.63\%  & 90.63\%                      & 90.63\%                      \\
\textbf{Tenor Trombone}                                  & Nonvib                         & 33.33\% & 33.33\%                     & 66.67\%                     & 78.79\%  & 78.79\%                      & 100.00\%                     \\
\multirow{2}{*}{\textbf{Trumpet}}                        & Nonvib                         & 51.43\% & 51.43\%                     & 82.86\%                     & 74.29\%  & 74.29\%                      & 97.14\%                      \\
                                                         & Vib                            & 51.43\% & 51.43\%                     & 82.86\%                     & 74.29\%  & 74.29\%                      & 100.00\%                     \\
\textbf{Tuba}                                            & Nonvib                         & 18.92\% & \multicolumn{1}{l}{-}       & \multicolumn{1}{l}{-}       & 48.65\%  & \multicolumn{1}{l}{-}        & \multicolumn{1}{l}{-}        \\
\multirow{2}{*}{\textbf{Viola}}                          & Pizz                           & 22.00\% & \multicolumn{1}{l}{-}       & \multicolumn{1}{l}{-}       & 28.00\%  & \multicolumn{1}{l}{-}        & \multicolumn{1}{l}{-}        \\
                                                         & Arco                           & 88.00\% & \multicolumn{1}{l}{88.00\%} & \multicolumn{1}{l}{91.00\%} & 100.00\% & \multicolumn{1}{l}{100.00\%} & \multicolumn{1}{l}{100.00\%} \\
\multirow{2}{*}{\textbf{Violin}}                         & Pizz                           & 25.27\% & \multicolumn{1}{l}{-}       & \multicolumn{1}{l}{-}       & 38.46\%  & \multicolumn{1}{l}{-}        & \multicolumn{1}{l}{-}        \\
                                                         & Arco                           & 92.22\% & \multicolumn{1}{l}{92.22\%} & \multicolumn{1}{l}{96.67\%} & 97.78\%  & \multicolumn{1}{l}{97.78\%}  & \multicolumn{1}{l}{100.00\%} \\ \bottomrule
\end{tabular}
\caption{Table showing the performance of the naive algorithm on monophonic samples from the University of Iowa Electronic Music Studios dataset, benchmarked against Noll's HPS algorithm. 1, 2, and 3 correspond to the whole dataset, sans outliers, and chroma accuracy respectively.}
\end{table}

\end{document}